\documentclass[a4paper,12pt]{article}

\usepackage{amsmath,amssymb,amsthm,amsfonts}
\usepackage{color}

\theoremstyle{plain}
\newtheorem{prop}{Proposition}[section]
\newtheorem{theorem}[prop] {Theorem}

\theoremstyle{definition}
\newtheorem{lemma}[prop]{Lemma}

\newtheorem{cor}[prop]{Corollary}

\theoremstyle{remark}
\newtheorem*{remark}{Remark}
\newtheorem*{example}{Example}

\newcommand{\N}{\mathbb{N}}
\newcommand{\R}{\mathbb{R}}
\newcommand{\Z}{\mathbb{Z}}
\newcommand{\C}{\mathbb{C}}
\renewcommand{\P}{\mathbb{P}}

\newcommand{\dd}{\mathrm{d}} 
\newcommand{\eps}{\epsilon}
\newcommand{\la}{\langle}
\newcommand{\ra}{\rangle}

\newcommand{\vect}[1]{\boldsymbol{#1}}

\DeclareMathOperator{\supp}{supp}
\DeclareMathOperator{\const}{const}

\DeclareMathOperator{\Tr}{Tr}

\title{Fermionic and bosonic Laughlin state \\
on thick cylinders}
\date{Sept. 18, 2011}
\author{Sabine Jansen\footnote{Weierstrass Institute for Applied Analysis and Stochastics, Leibniz Institute in Forschungsverbund Berlin e.V., Mohrenstr. 39, 10117 Berlin, \texttt{jansen@wias-berlin.de}.}}

\begin{document}

\maketitle

\begin{abstract}
	We investigate a many-body wave function for particles on a cylinder known as Laughlin's function. 
	It is the power of a Vandermonde determinant times a Gaussian. Our main result is:  in a many-particle limit,
	 at fixed radius, all correlation functions have a unique limit, and the limit state  
	has a non-trivial period in the axial direction. 
	The result holds regardless how large the radius is, for fermions as well as bosons.  
	In addition, we explain how the algebraic structure used in proofs relates to a ground 
	state perturbation series and to quasi-state decompositions, and we show that the monomer-dimer function introduced in an earlier work is an exact, zero energy, ground state 
	of a  suitable finite range Hamiltonian; this is interesting because of formal analogies with 
	some quantum spin chains. \\
	
\end{abstract}

\tableofcontents

\section{Introduction}

In this article we study a many-body wave function for  particles on long cylinders.
The wave function is the product of the power of a Vandermonde determinant and 
a Gaussian. It arises as a model wave function in the fractional Hall effect~\cite{lau83} and is 
known as Laughlin's wave function, but the wave function, or  variants of it, 
  play a role in other areas too, e.g., rapidly rotating Bose gases (see~\cite{lewin-seiringer} and the references therein)
and classical Coulomb systems 
(``jellium'' or one-component plasma), see for example~\cite{samaj,forr06}. The function also resembles expressions studied in 
 random matrix theory. 

A full discussion of the quantum Hall effect background is beyond the scope of this article, and we should stress that our result does not seem to have direct implications in terms of quantized conductance. There are, nevertheless, reasons that might make the result interesting in the quantum Hall effect context. The key words are Laughlin's argument, Chern numbers, incompressibility, topological order, and the plasma analogy. 

 Laughlin's argument for the integer Hall effect~\cite{lau81} used the cylinder geometry and gauge invariance in order to show that Hall conductances should have integer quantization. It was later suggested that a ground state degeneracy, possibly due to translational symmetry breaking, is required in order to reconcile this argument with fractional conductances~\cite{taowu}. The relation between ground state 
degeneracy and fractional conductance can in fact be made rigorous on a torus, in the Chern number approach~\cite{avron-seiler-yaffe}. 
This  leads to the question whether Laughlin's state, on a cylinder, considered as an approximate ground state, is degenerate or not:  this was the initial motivation of the present article, see~\cite{jls} for further references. Our main result is that at filling factor $1/p$, Laughlin's state exhibits indeed $p$-fold symmetry breaking, at all values of the radius, complementing  previous thin cylinder results~\cite{rh,jls}.

This shows, in a way, that the use of Laughlin's function is consistent with an important ingredient of the theory of the fractional Hall effect, \emph{incompressibility} or the existence of a gap above the ground state(s),  see the review~\cite{bf} and also the discussion in~\cite{bes}.
Indeed, on a cylinder, at non-integer filling factor, incompressibility implies translational symmetry breaking~\cite{koma}. Let us mention, however, that our symmetry breaking result holds at all values of the radius and of the filling factor, regardless of whether or not there is a gap.

The symmetry breaking proven here can be read as another illustration of the geometry-dependent degeneracy which is considered a hallmark of ``topological order''. We do not wish to discuss this notion here, and instead refer the reader to~\cite{hastings-mich} and the references therein for a discussion in combination with the fractional Hall effect. 

The plasma analogy refers to the observation that the modulus squared of Laughlin's function is proportional to the Boltzmann weight for a classical system of point charges moving in a neutralizing background. It was originally invoked~\cite{lau83} in order to justify that, at not too low filling factors, in a disk geometry, Laughlin's state describes a homogeneous liquid. Interestingly, this very same argument, when adapted to the cylinder geometry, suggests 
a spatial periodicity. From this point of view, the symmetry breaking proven here is closely related to results for one-dimensional jellium~\cite{kunz,brascamp-lieb,aizenman-martin}, 
jellium on a semi-periodic strip at the ``free-fermion point''~\cite{cfs} and at even-integer coupling~\cite{swk}, and jellium on quasi one-dimensional tubes with periodic or Neumann boundary conditions, at arbitrary coupling~\cite{ajj}. 

A curious aspect of the quasi one-dimensional jellium is that it interpolates between one-dimensional jellium, which is known to have no phase transition~\cite{kunz}, and higher-dimensional jellium, which is expected to undergo a Wigner crystallization phase transition (related to the formation of vortices in a rotating Bose gas). An interesting open question is, therefore, whether at  low filling factor $1/p$, Laughlin's state has a phase transition as the radius is  varied. We know that  for every fixed filling factor, when the cylinder radius is sufficiently small, the infinite volume correlations depend on the radius in an analytic way, and 
there is exponential clustering~\cite{jls}. 
 Our results imply that if there is a phase transition, it cannot manifest itself in a change of the state's spatial periodicity. A possible scenario, instead, could be that exponential clustering is replaced with algebraic decay of correlations, which in turn is related to the question whether the gap of some toy Hamiltonian vanishes at some finite value of the radius as the radius is increased. 

The simplest motivation, perhaps, is to consider Laughlin's function as a partially solvable, quantum many-body toy problem. It is well-known, indeed, that the wave function is an exact ground 
state of a suitable interaction
which encodes 
that the wave function has zeros of a given order as particles get close~\cite{haldane,pt,tk}. 
 Our problem can also be seen as a problem for fermions, or bosons, 
on a one-dimensional lattice, 
and the interaction takes the form  
\begin{equation} \label{eq:ham}
	H = \sum_{k_1+k_2 = n_1 + n_2} F(n_1 -n_2) F(k_1 - k_2)  c_{k_1}^* c_{k_2}^* c_{n_2} c_{n_1}
\end{equation}
for some suitable rapidly decaying function $F$. Studying Laughlin's wave function amounts to studying the ground 
state of this simple looking Hamiltonian. \\

Our main results are the following: when we let the particle number go to infinity, 
at fixed radius, all correlation functions have a unique limit, and the limit state has a non-trivial period along the cylinder axis. 
At filling factor $1/p$, the period is $p$ times the period of the filled lowest Landau level; equivalently,  $p$ times the period of the Hamiltonian
$H$ from Eq.~\eqref{eq:ham}.   We also show that the state is clustering, and that bulk correlation functions are insensitive to the 
precise choice of the domain of integration. This is akin to the accumulation of excess charge at the boundary of Coulomb systems. 

These results extend  previous results for fermions on thin cylinders~\cite{jls}. They leave open, however, some questions that 
were answered positively on thin cylinders. Most notably, on thick cylinders, we do not know whether the symmetry breaking 
is already apparent at the level of the one-particle density, and we do not know whether the correlation length is finite. 
This last question is, via the exponential clustering theorem~\cite{ns} related to another open question, namely, whether 
the Hamiltonian $H$ has a gap above its ground state (see, however, \cite{sfl} for numerical results at filling factor $1/3$).

The symmetry breaking proven here is  also  closely related to results for classical quasi one-dimensional Coulomb systems~\cite{ajj}. From this perspective our result 
is, on the one hand, a specialization of~\cite{ajj} to  $1+1$-dimensional semi-periodic strips; on the other hand, our result is an improvement 
in the sense that we examine the full quantum mechanical state and we prove uniqueness (up to shifts) and ergodicity 
of the limiting state.

Our proofs follow~\cite{jls} by exploiting the algebraic structure of the wave function. 
Some of the results in~\cite{jls} were proven using a perturbative arguments in the thin strip limit. 
We adapt methods from~\cite{ajj} in order to show that part of these results are not perturbative  
but true for all radius values. In particular, the associated renewal process has always finite mean. \\

Another goal of this article is to cast the algebraic manipulations and expressions of infinite volume 
correlation functions from~\cite{jls} in a form that allows for a more explicit 
expression of the relationship between quantum mechanics and an associated probability problem. 
We will see that the formalism implicitly used in~\cite{jls} does not stand alone, but is actually similar
to structures encountered in  magnetic itinerant electrons and in spin chains.

Quickly summarized, \cite{jls} exploited an  algebraic property 
of powers of Vandermonde determinants in order to give a representation both of the  
normalization $C_N=||\Psi_N||^2$ and the correlation functions. The normalization was expressed 
as a polymer partition function,  or more precisely, a sum over partitions of some discrete volume into non-overlapping 
intervals with multiplicative weights.

This representation of the normalization is similar to a formula 
for the canonical partition function for magnetic itinerant electrons given by 
Aizenman and Lieb, Eq.~(11) in~\cite{al}. The canonical partition function 
was written there as a sum over partitions $\{n_i\}$ of the number of electrons $N$,
with the interpretation that 
	\begin{quote}
	``as far as the $z$-component $S$ is concerned (...) the system is in a  superposition of states in which the particles form independent `cliques'
	of size $n_i$.''  \hfill {\small \cite{al}, Remark~(1) after Eq.~(11)}
	\end{quote}
The clique structure carries over to a representation of the correlation functions:
for quantum spin chains, the state can is a weighted superposition of (quasi-)states, corresponding each 
to a partition of spins into ``random clusters of total spin zero'' (\cite{an}). 
The weights  define a probability measure on partitions, and a good 
understanding of this probability measure is  a starting point for a good understanding of the quantum-mechanical state.

Something similar can be done for Laughlin's wave function.  The physical meaning of the analogous clique structure is 
unclear, but we shall try to give some intuition in terms of a ground state perturbation series. 
Let us explain right away 
why the representation is so useful for the infinite volume limit. We  will write expectations as convex 
combinations of linear functionals (\emph{quasi-states}), indexed by partitions $\mathcal{X}$:
\begin{equation*} 
	\frac{ \la \Psi_N, a \Psi_N\ra }{||\Psi_N||^2} = \sum_{\mathcal{X}} p_N(\mathcal{X}) \omega_\mathcal{X} (a).
\end{equation*}
The weights define a probability measure $\P_N$ on partitions. The decomposition allows us to 
reduce the question of convergence of quantum mechanical states to the much simpler problem of convergence of 
probability measures. In the limit $N\to \infty$, the measure $\P_N$ has a limit $\P$, and 
\begin{equation*}
	\frac{ \la \Psi_N, a \Psi_N\ra }{||\Psi_N||^2} \to \int \dd \P(\mathcal{X}) \, \omega_\mathcal{X} (a).
\end{equation*}
The measure $\P$ is a $p$-periodic renewal process and it is mixing. The infinite volume state inherits the periodicity 
and the clustering from the measure $\P$.  \\

Finally, we show that the solvable monomer-dimer function introduced 
in~\cite{jls} arises as an exact ground state. At filling factor $1/3$, the Hamiltonian is obtained 
from $H$ by restricting the interaction range to $|k_1 - k_2|\leq 3$, $|n_1 - n_2|\leq 3$. 
This does not enter the proofs of our main results in any way, but is interesting for two reasons. 
First, the monomer-dimer Hamiltonian might be a good toy model for open  questions about gaps and incompressibility. 
Second, the monomer-dimer Hamiltonian, as a sum of non-commuting, positive, local operators $\sum_s B_s^* B_s,$ 
resembles spin chains with nearest neighbor spin singlet projections. This calls for a better understanding of the relation 
of our monomer-dimer structure and the valence bond structures and dimerization encountered in spin chains~\cite{aklt,an}. \\

The article is organized as follows: Sect.~\ref{sec:results} states the main results in terms of the continuum wave function; 
the lattice formulation, many-body Hamiltonian and monomer-dimer model are presented in Sect.~\ref{sec:lattice}. 
The renewal structure of the wave function and the quasi-state decomposition are explained in Sect.~\ref{sec:renewal}. 
In Sect.~\ref{sec:peierls}, we explain how methods from quasi-1D Coulomb systems help strengthen perturbative results from~\cite{jls}. 
Finally, Sect.~\ref{sec:correlations} weaves together the formalism from Sect.~\ref{sec:renewal} and the bounds from Sect.~\ref{sec:peierls} 
to prove our results on correlation functions.

\section{Main results} \label{sec:results}

Consider $N$ particles moving on a cylinder of radius $R>0$. Particles have complex coordinates $z = x + i y$. 
 The real part $x$ is a coordinate along the cylinder axis, and the imaginary part 
$0 \leq y \leq 2\pi R$ is an angular coordinate around the cylinder. We are interested in the wave function 
\begin{equation}
	\Psi_N(z_1,...,z_N) = \kappa_N 
		 \prod_{1 \leq j<k \leq N} 
		\Bigl(\exp( z_k/R) - \exp( z_j/R) \Bigr)^p \exp\left( - \frac{1}{2\ell^2} \sum_{k=1}^N x_k^2 \right).
\end{equation}
The parameters $\ell>0$ and $p\in \N$ are considered fixed. In the quantum Hall effect, $\ell \propto 1/ \sqrt{|B|}$ is the magnetic length, 
and $1/p$ is the filling factor.  When $p$ is odd, $\Psi_N$ is a wave function for fermions; when $p$ is even, $\Psi_N$ is a wave function for 
bosons. The multiplicative constant $\kappa_N$ is irrelevant for correlation functions; a convenient choice is nevertheless given in Eq.~\eqref{eq:kappan}. 

Before we state our results, we mention an illuminating alternative form of the wave function's modulus.  
Using the factorization 
\begin{equation*}
	\exp(z_k/R) - \exp(z_j/R) = \Bigl(1- \exp([z_j-z_k]/R) \Bigr) \exp(z_k/R). 
\end{equation*}
and a completion of squares, we find 
\begin{multline} \label{eq:wave-fct-period}
	|\Psi_N(z_1,\ldots, z_N)|^2 \propto \prod_{1 \leq j < k \leq N} \Bigl| 1- \exp\Bigl(\frac{z_j -z_k}{R}\Bigr)
		 \Bigr|^{2p} \\
		\times \exp\left( - \frac{1}{\ell^2} \sum_{k=1}^N \Bigl(x_k - (k-1) \frac{p \ell^2}{R}\Bigr)^2 \right). 
\end{multline}
The Gaussian clearly favors a period $p \ell^2/R$ in the $x$-direction. The pre\-factor should not 
destroy the periodicity, since 
in the sector $x_1 \leq x_2 \leq \cdots \leq x_N$, 
$$ 1- \exp\frac{z_j-z_k}{R} = 1 + O \left( \exp\Bigl(- \frac{|x_j-x_k|}{R} \Bigr) \right)$$
is almost unity when $z_j$ and $z_k$ are far apart. 
Moreover, since the Gaussian decays exponentially outside a finite cylinder of length $p N \ell^2/R$, 
the exact choice of the domain of integration $\Lambda$
should not affect bulk correlation functions. 
Theorems~\ref{thm:thermolim} and~\ref{thm:symbreak} below state that this picture is correct; that this is true for all values of the cylinder radius, and not just for a radius small compared to the magnetic length $\ell$, is in essence the main result of this paper. 

\begin{remark} The dimensionless parameters are $p\in \N$ and the ratio $\gamma:= \ell/R$. 
In the proofs, we will choose length units such that $\ell =1$, so that the period of the state becomes $p \ell^2 /R = p \gamma$. 
\end{remark}

We now state our results in more detail. 
Let $\mathcal{Z} = \R \times [0,2\pi R]$. For a given domain of integration $\Lambda \subset \mathcal{Z}$ 
and $N\in \N$, we consider the $n$-point correlation function 
\begin{multline*}
 	\rho_n^{N,\Lambda} (z_1,\ldots,z_n;z'_1,\ldots,z'_n) \\
		:= \frac{N!}{(N-n)!}
		\,  \frac{1}{||\Psi_N||^2_\Lambda} \,
	  \int_{\Lambda^{N-n}} 
		 \Psi_N(z_1,\ldots,z_n,z_{n+1},\ldots,z_N) \\
		\times \overline{\Psi_N(z'_1,\ldots,z'_n,z_{n+1},\ldots,z_N)}\, \dd x_{n+1} \dd y_{n+1} \cdots \dd x_N \dd y_N. 
\end{multline*}
We willl consider domains of the form 
\begin{equation} \label{eq:domchoice}
  \Lambda = [a,b]\times [0,2\pi R],\quad - \infty \leq a \leq C,\quad N p \ell^2/R - C \leq b \leq \infty
\end{equation}
for some $N$-independent 
constant $C$, but $a$ and $b$ possibly $N$-dependent. For example, we can integrate over the infinite cylinder 
$\mathcal{Z}:\, -\infty \leq x \leq \infty$, the semi-infinite cylinder $x\geq 0$, or a finite cylinder of length 
$p N \ell^2/R$. 

We are interested in the limit $N\to \infty$ at fixed cylinder radius $R>0$, and fixed parameters 
$\ell, p$. The cylinder length $b-a$ as in Eq.~\eqref{eq:domchoice} goes to infinity. In this limit, 
the bulk correlation functions have a unique limit, and the limit is $p\ell^2/R$-periodic. More exactly, 
 we shall see that  
\begin{multline} \label{eq:periodic}
  \Bigl( \prod_{j=1}^N u(py_j)\Bigr)  \, \rho_n\Bigl( \{z_j- p \ell^2/R\}_{j=1}^n;\{z'_j- p \ell^2/R\}_{j=1}^n\Bigr)\,
	 \prod_{j=1}^N \overline{u(p y'_j)} \\
      = \rho_n(z_1,\ldots,z_n;z'_1,\ldots,z'_n).
\end{multline}
for all $z_1,\ldots, z'_n \in \mathcal{Z}$, and $u(py) = \exp( i p y /R)$; the unitary $u(y)$ comes from 
 the \emph{magnetic translation}
$$\Bigl( t( \ell^2/R\, \vect{e}_x) \psi\Bigr)(z) = \exp(i  y /R) \psi(z-  \ell^2/ R).$$

\begin{theorem}[Infinite volume correlation functions, $p$-periodicity] \label{thm:thermolim}
  Let $N\to \infty$ at fixed $R$ and $\Lambda$ as in Eq.~\eqref{eq:domchoice}.  Then there is  
	 a unique family of  infinite volume correlation functions $\rho_n(\{z_j\};\{z'_j\})$ such that, for suitable  
	 $\varepsilon(d)$ with $\varepsilon(d) \to 0$ as $d\to \infty$, and suitable constant $K_n$, 
 \begin{equation*}
    \Bigl| \rho_n^{N,\Lambda} (z_1,\ldots,z_n;z'_1,\ldots,z'_n) - \rho_n(z_1,\ldots,z_n;z'_1,\ldots,z'_n) \Bigr| \leq 
      K_n \varepsilon(d), 
  \end{equation*}
 	whenever $z_1,\ldots,z'_n \in \Lambda$ satisfy $x_j, x'_j \geq d$  and 
	$x_j, x'_j \leq pN \ell^2/R - d$ for all $j$.  
	The infinite volume correlations are $p\ell^2/R$-periodic in the sense of Eq.~\eqref{eq:periodic}. 
	They do not depend on the precise choice of $\Lambda$.  
  \end{theorem}
The error terms $K_n$ and $\varepsilon(d)$  may depend on $C$ in the choice of $\Lambda$, Eq.~\eqref{eq:domchoice}. The infinite volume 
correlation functions are invariant (in the usual sense) with respect to  arbitrary shifts in the $y$-direction, $y_j \to y_j +s$, and with respect to 
	    reversals $z_j \to - z_j$. 

\begin{remark}
	 The correlation functions $\rho_n$ determine uniquely a state $\la \cdot \ra$ on a 
	suitable bosonic 
	or fermionic observable algebra, and we also have convergence of states on that algebra, 
	    i.e., $\la \Psi_N, a \Psi_N \ra / ||\Psi_N||^2 \to \la a \ra$. 
\end{remark}

The periodicity statement in the previous theorem does not yet guarantee that $p \ell^2/R$ is the 
smallest period of the state, but the next theorem does.

\begin{theorem}[Symmetry breaking] \label{thm:symbreak}
  For all $\ell, R>0$ and $p \in \N$, there is some correlation function $\rho_n(z_1,\ldots,z_n; z_1,\ldots,z_n)$ 
  that  has  $p \ell^2/R$  as its smallest period. 
\end{theorem}

This should be contrasted with the result from~\cite{jls}: on thin cylinders ($R \ll \ell$),  we know 
that $p \ell^2/R$ is the smallest period of the one-particle density $\rho_1(z;z)$, which is 
a function of $x$ alone. Theorem~\ref{thm:symbreak} leaves open whether this stays true for thick cylinders. In any case, however, the one-particle density has a non-trivial period which must be a multiple of the filled lowest Landau period $\ell^2/R$~\cite[Lemma 2]{jls}.

The use of higher order correlation functions in the formulation of symmetry breaking may look 
confusing, so let us try to make it more concrete. 
Suppose we are given  samples of 
electron gas in the state $\Psi_N$, and we can make repeated measurements of the number of 
particles $N(x)$, $N(\tilde x)$ in small width $\epsilon$ annuli at two different abscissas $x$ and $\tilde x$. 
This gives two histograms which will look different unless $\tilde x -x$ is a multiple of $p\ell^2/R$.
On thick cylinders, the histograms might  have the same average $\la N(x) \ra = \la N(\tilde x)\ra$
but different shapes, so  we might have to look at higher moments $\la N(x)^m\ra$ in order to see the difference.

The $C^*$ algebraic content of Theorem~\ref{thm:symbreak} is that the states 
$\omega_1,\ldots,\omega_p$ associated with the bulk correlation functions $(\rho_n)$ 
and its shifts by $\ell^2 /R$  along the cylinder axis are distinct. We shall see that they are 
actually \emph{disjoint} in the sense of~\cite[Sect. 4.2.2, p. 370]{bratteli-robinson1}, see Sect.~\ref{sec:symbreak}. 
This corresponds to the mutual singularity of probability measures proven in \cite{ajj}. 

The correlation functions are clustering. We formulate the next theorem for integration on the infinite 
cylinder $\mathcal{Z}$, but note that the insensitivity to the precise domain of integration stated in Theorem~\ref{thm:thermolim} 
allows us to transfer the statement to other domains of integration.  
\begin{theorem}[Clustering] \label{thm:clustering}
	Let $m,n\in\N$ and $\{z_j,z'_j;\ 1\leq j \leq m\}$, $\{z_j,z'_j:\ m+1 \leq j \leq n\}$ 
	be coordinate clouds having distance $\geq d>0$ to each other. Then 
	  \begin{equation*}
		\Bigl| \rho_{m+n}^{N,\mathcal{Z}} (z_1,\ldots, z'_{m+n} )- 
		      \rho_m^{N,\mathcal{Z}} (z_1,\ldots,z'_m)\,\rho_n^{N,\mathcal{Z}} (z_{m+1},\ldots,z'_{m+n})\Bigr| 
		    \leq K_{m+n} \varepsilon(d).
	  \end{equation*}
	for some $N$-independent function  $\varepsilon(d)$ with 
	  $\varepsilon(d) \to 0$ as $d\to \infty$, and some suitable constant $K_{m+n}$. 
\end{theorem}
On thin cylinders, it was shown in~\cite{jls} that there is exponential clustering, i.e., 
$\varepsilon(d) \to 0$ exponentially fast. Whether this stays true on thick cylinders is an open question. 

Clustering is not only of interest in itself, but also enters the proof of the insensitivity of bulk correlation functions to the precise 
choice of domain of integration. We will see, indeed, that a switch from integration on the infinite cylinder to integration on a finite
cylinder can be seen as a perturbation localized at the boundary; clustering allows us, then, to show that 
this boundary perturbation does not affect bulk correlations. 

 Theorem~\ref{thm:clustering} is a statement about diagonals: in each of the coordinate clouds, the number of primed 
variables equals the number of unprimed variables -- we do not separate $z_j$ from $z'_j$. 
There is off-diagonal decay too, but do not offer any explicit estimate, except for 
the one-particle matrix.

\begin{prop}[Off-diagonal decay of the one-particle matrix] \label{prop:off-diagonal}
	For some suitable constant $K$ and all $N$, $z$, $z'$, 
	\begin{equation*}
		\bigl| \rho_1^{N,\mathcal{Z}} (z;z') \bigr| \leq K \exp\Bigl(  - \frac{(x-x')^2}{4 \ell^2} \Bigr).
	\end{equation*}		
\end{prop}
Prop.~\ref{prop:off-diagonal} is proven in Sect.~\ref{sec:bosons}. 

\section{Interacting fermions and bosons} \label{sec:lattice}

For the proofs, it is convenient to view the wave function as a wave function for 
fermions, or bosons, on a one-dimensional lattice. Here we explain how this is done. 
In addition, we recall the expression of the many-body Hamiltonian whose exact ground state Laughlin's function is, and we show that the monomer-dimer function from~\cite{jls} is the exact ground state of a suitable Hamiltonian too. 

\paragraph{One-dimensional lattice} Straightforward algebra shows that $\Psi_N(z_1,\ldots,z_N)$ is a linear combination of products of one-particle functions 
\begin{equation}
\begin{aligned} \label{eq:basis}
  \psi_k(z) & \propto \exp( k z/ R) \exp\Bigl( -\frac{ 1}{2 \ell^2} x^2\Bigr) \\
		& \propto \exp( i k y / R) \exp\Bigl( - \frac{1}{2\ell^2} (x- k \ell^2/R)^2\Bigr)
\end{aligned}
\end{equation}
with $0 \leq k \leq pN - p$ ; see Eq.~\eqref{eq:expansion} below. The proportionality constants are chosen positive and such that the $\psi_k(z)$, $k\in \Z$, are
orthonormal in $L^2(\mathcal{Z})$. 

The one-particle function $\psi_k(z)$ is localized in the cylinder axis direction around $x = k \ell^2/R$. 
Thus we may identify  lattice sites $k \in \Z$ with the orbitals $\psi_k(z)$ and view   
our wave function as a function for bosons, or fermions, on a one-dimensional lattice. 

It is convenient, therefore, to work with creation and annihilation operators $c_k^*$ and $c_k$, whose definition 
we briefly recall. When $p$ is odd, let $\mathcal{F}$ be the fermionic Fock space 
associated with the one-particle space $L^2(\mathcal{Z})$. 
 For $k \in \Z$, the creation operator $c_k^*$ is the linear map in $\mathcal{F}$ given by $c_k^* F = \psi_k \wedge F$, 
with the  antisymmetrized product
 $f\wedge g (z,w) = [f(z) g(w) - g(z) f(w)]/\sqrt{2}$.  The annihilation operator $c_k$ is the adjoint of $c_k^*$ in $\mathcal{F}$. 
When $p$ is even, the definitions are similar, except the wedge product is replaced by a symmetrized tensor product 
and we have to work in bosonic Fock space.  Thus $c_k^*$ creates a particle at lattice site $k$, or in the orbital $\psi_k(z)$, and $c_k$ annihilates a particle.
The number of particles at site $k$ is given by the occupation number operator $\hat n_k:= c_k^* c_k$. 

We are going to investigate expectations 
\begin{equation*}
	\la a \ra_{N,\Lambda} := \frac{\la \Psi_N, a \Psi_N\ra_\Lambda}{||\Psi_N||_{\Lambda}^2}.
\end{equation*}
for observables $a$ that are products of creation and annihilation operators $c_k^*$ and $c_k$. 
The continuum $n$-point correlation functions is easily recovered from 
expectations of such operators. For example, the one-particle function is
\begin{equation*}
	\rho_1^{N,\Lambda}(z;z') 
		= \sum_{k,m\in \Z} \la c_m^*c_k \ra_{N,\Lambda} \psi_k(z)\, \overline{\psi_m(z')}
\end{equation*}
 In our case, the expectations $\la c_m^*c_k\ra_{N,\Lambda}$ actually vanish when $k \neq m$.

When $p=1$, the wave function is a simple Slater determinant and describes independent fermions. As soon as $p \geq 2$, it has built-in 
correlations and is the ground state of a suitable many-body Hamiltonian. 

\paragraph{Many-body Hamiltonian}  
The wave function $\Psi_N$ at filling factor $1/p$ has  a zero of order $p$ when two particles get close, $z_j - z_k \to 0$.  In the subspace
spanned by products of one-particle functions $\psi_0(z),\ldots, \psi_{pN-p}(z)$, $\Psi_N$  is the unique 
wave function with this property. This leads to the characterization of $\Psi_N$ as the unique ground state 
of a suitably defined Hamiltonian~\cite{haldane,pt,tk}. 
The aim of this section is to recall the form of this Hamiltonian 
for the cylinder problem, see~\cite{lee-leinaas,sfl} for similar expressions at filling factor $1/3$.
 This will not be used in any of our later proofs, but helps the interpretation of our results in terms of  a quantum many-body problem. 

Let $H_n(z)$, $n\in \N_0$,  be the Hermite polynomials, given by the generating function relation
\begin{equation*}
	\exp( 2 t z - z^2)  = \sum_{n=0}^\infty \frac{t^n}{n!} H_n(z)
\end{equation*}
and $F_n(t):= H_n(t) \exp( - t^2 /4)$. At filling factor $1/p$, let  $F(t):= F_0(t)+\cdots + F_{p-1}(t)$; 
alternatively, sum only those $F_k(t)$ where $0 \leq k \leq p-1$ and $k$ has the same parity as $p$. 
For example,  $F(t) = F_1(t)$ when $p=3$. 
Consider the formal sum  over $k_1,k_2,n_1,n_2 \in \Z$ with $k_1+k_2 = n_1 +n_2$: 
\begin{equation} \label{eq:formal-1}
	H = \sum_{k_1 + k_2 = n_1 + n_2  } F\bigl( (n_1-n_2) \gamma \big) F\bigl((k_1-k_2) \gamma\bigr)
		c_{k_1}^* c_{k_2}^* c_{n_2} c_{n_1}.
\end{equation}
(Recall $\gamma = \ell/R$.)
We 
define the finite volume Hamiltonians via free boundary conditions: 
For $\mathcal{L} = \{0,\ldots, pN-p\}$ we let $H_\mathcal{L}$ be the Hamiltonian obtained from this formal 
sum by restricting summation to integers $k_1,k_2,n_1,n_2 \in \mathcal{L}$. 
\begin{prop} \label{prop:ground-state}
	The Hamiltonian $H_\mathcal{L}$ is positive $H_\mathcal{L} \geq 0$, $\Psi_N$ is an exact ground state $H_\mathcal{L} \Psi_N =0$, and $\Psi_N$ is the unique  
	ground state of $H_\mathcal{L}$, considered as a Hamiltonian for $N$ particles in the finite lattice 
	$\mathcal{L} = \{0,1,\ldots, pN-p\}$.  	
\end{prop}
The content of the Proposition seems is well-known in the physics literature. For mathematical completeness, we include nevertheless a proof adapted to the cylinder setting.

\begin{proof}[Proof]
	The positivity of the Hamiltonian follows from an alternative expression of Eq.~\eqref{eq:formal-1}
	\begin{equation}\label{eq:formal-2}
		H  = \sum_{ s \in (1/2) \Z} B_s^* B_s, \quad 
			B_s  = \sideset {}{'}\sum_{k} F(2k\gamma) c_{s-k} c_{s+k}.
	\end{equation}
	see~\cite[Eq. (2)]{lee-leinaas}.
	The first sum is over half-integers $s$, the second sum is either over integers $k\in \Z$ (when $s$ is integer) 
	or half-integers $k \in ((1/2) \Z)\backslash \Z$ (when $s$ is half-integer but not integer); note that the sum of two integers and their difference have the same parity, 
	 $2 s = k_1 + k_2 = k_1 - k _2 \mod 2$.  
	For the ground state property, let us first look at two-particle functions. Let 
	\begin{equation*}
		\Psi(z_1,z_2) = \sum_{m_1,m_2} \lambda(m_1,m_2) \psi_{m_1}(z_1) \psi_{m_2}(z_2)
	\end{equation*}
	with the normalized one-particle functions $\psi_m(z) \propto \exp( - m^2 \gamma^2/2) \exp( i m\gamma y - x^2 /2)$. 
	We use  the generating functions relations for the Hermite polynomials: 
	and expand 
	\begin{equation*}
		\exp( x^2 - z^2 ) \Psi(Z+z,Z-z) = \sum_{n=0}^\infty \frac{z^n}{n!} \Phi_n( Z). 
	\end{equation*}
	where  
	\begin{equation*}
		\Phi_n(Z)  \propto \sum_{m_1,m_2} \lambda(m_1,m_2)  F_n\bigl( (m_1- m_2)\gamma \bigr)
					\psi_{m_1+m_2}(Z).
	\end{equation*}
	Looking at the $L^2$ norm of $\Phi_n(Z)$ we see that  $\Phi_n(Z)\equiv 0$ if and only if 
	\begin{equation} \label{eq:vanish}
		\sum_{m_1+m_2=k_1+k_2} \overline{\lambda(m_1,m_2)} F_n\bigl(( m_1 - m_2)\gamma \bigr) F_n\bigl(( k_1 - k_2)\gamma \bigr) \lambda(k_1,k_2) = 0. 
	\end{equation}
	 $\Psi(z_1,z_2)$ has a zero of order $\geq p$  as two particles get close if and only if $\Phi_0(Z),\ldots,\Phi_{p-1}(Z)$  vanish 
	identically, which because of Eq.~\eqref{eq:vanish} happens if and only if $\la \Psi, H_\mathcal{L} \Psi \ra =0$. Note that, depending on the parity of $p$, 
	symmetry or antisymmetry gives $\Phi_{2n}(Z)\equiv 0$ or $\Phi_{2n+1}(Z)\equiv 0$ for free. This proves Prop.~\ref{prop:ground-state} for $N=2$ particles. For more than $2$ particles, the previous argument is applied separately for each pair of particle coordinates $(z_j,z_k)$, $j<k$. 	 
\end{proof}

\paragraph{Monomer-dimer function} 
In~\cite[Sect. 2]{jls} a solvable model was introduced; this was motivated, roughly, by neglecting the overlap between next-nearest Gaussians, or replacing Gaussians with functions of compact support. Here we explain an alternative motivation: it turns out that the  monomer-dimer function is an exact ground state of a Hamiltonian obtained from $H$ by keeping only leading order terms.

 We only look at $p=3$.  We start as in Eq.~\eqref{eq:formal-1}
but keep only terms with $|k_1 - k _2| \leq 3$ and $|n_1 - n_2 |\leq 3$. This gives, up to a factor 
$\exp( - \gamma^2/2)$, the \emph{monomer-dimer Hamiltonian}
\begin{multline*}
	 H^\mathsf{MD} = \sum_k \Bigl( \hat n_k \hat n_{k+1} +  4 \exp(- 3 \gamma^2/2) \hat n_k \hat n_{k+2} 
			+ 9 \exp(- 4 \gamma^2) \hat n_k \hat n_{k+3} \\
				+ 3 \exp(- 2 \gamma^2) (c^*_{k+1} c^*_{k+2} c_{k+3} c_k + c^*_{k} c^*_{k+3} c_{k+2} c_{k+1} ) \Bigr).
\end{multline*}
Equivalently, 
\begin{multline} \label{eq:mondim}
	H^\mathsf{MD} = 4 \exp(- 3 \gamma^2/2) \hat n_0 \hat n_2 \\
		+ \bigl(c^*_1 c^*_2 + 3 \exp( - 2 \gamma^2) c^*_0 c^*_3 \bigr) \bigl(c_2 c_1 + 3 \exp(- 2 \gamma^2)c_3 c_0 \bigr) 
		+\text{ translates}.
\end{multline}
These sums are, of course, formal; given a finite lattice $\mathcal{L} = \{0,1,\ldots, 3N-3\}$ we associate as before 
a Hamiltonian $H_\mathcal{L}^\mathsf{MD}$ via free boundary conditions. 

Note that in a thin strip (large $\gamma$) limit, truncating the interaction range amounts to keeping only leading 
order contributions, i.e.,  we neglect high powers of $\exp( - \gamma^2)$. 

We define a many-body wave function as follows: for $k \in \Z$, let 
\begin{equation*}
	A_{\{k\}} := c_{3k}^*,\qquad A_{\{k,k+1\}} := - 3\exp(-2\gamma^2) c^*_{3k+1}c^*_{3k+2}.
\end{equation*}
 Thus 
a monomer operator creates a single particle and a dimer operator creates a pair of particles. 
The wave function is a sum  over monomer-dimer partitions $(X_1,\ldots,X_D)$ of $\{0,1,\ldots,N-1\}$ (labeled from left to right):
\begin{equation*}
	\Psi_N^\mathsf{MD}:= \sum_{(X_1,\ldots, X_D)} A_{X_1}\cdots A_{X_D}  |\text{vacuum} \ra.
\end{equation*}
$\Psi_N^\mathsf{MD}$ is the sum of $|0\,3\,6\cdots (3N-3)\ra$, corresponding to a partition consisting  of monomers only, 
and orthogonal terms obtained by hopping one or several  pairs of particles, for example, $ 3\,6 \to 4\,5$. 	In particular, 
all particle pairs have distance $1$, $3$ or higher; there are no pairs of particles at mutual  distance~$2$.

We have an analogue of Prop.~\ref{prop:ground-state}. We do not, however, prove uniqueness of the ground state. 
\begin{prop} \label{prop:ground-mondim}
	The Hamiltonian $H^\mathsf{MD}_\mathcal{L}$ is positive, $H_\mathcal{L}^\mathsf{MD} \geq 0$, and $\Psi_N^\mathsf{MD}$ is an exact, zero energy, ground state. 
\end{prop}

\begin{proof}
	The positivity follows from Eq.~\eqref{eq:mondim}. In order to check that $\Psi_N^\mathsf{MD}$ is a zero energy state, we first note that 
	$\Psi_N^\mathsf{MD}$ has no pairs of particles at mutual distance $2$ and therefore $\hat n_j \hat n_{j+2} \Psi_N^\mathsf{MD}=0$ for all $j$. 
	It remains to see that 
	$$ B_j \Psi_N^\mathsf{MD} =0, \quad B_j := c_{j+1} c_{j+2} + 3 e^{- 2 \gamma^2} c_j c_{j+3}  $$ 
	for all $j$. When $j$ is not a multiple of $3$, $B_j \Psi_N^\mathsf{MD}$ vanishes because the wave function has no particle pairs 
	at positions $(j,j+3)$ or $(j,j+1)$.  
	When $j = 3k$ for some $k= 0,\ldots, N-1$,
	\begin{equation} \label{eq:bjzero}
		\bigl(c_{3k+1} c_{3k+2} + 3 e^{- 2 \gamma^2} c_{3k} c_{3k+3} \bigr) A_{X_1} \cdots A_{X_D} |\text{vacuum}\ra \neq 0 
	\end{equation}
	implies that either the partition has a dimer $\{k,k+1\}$ or it has two monomers $\{k\}$, $\{k+1\}$. As a consequence we can factorize 
	\begin{equation*}
		B_{3k} \Psi_N^\mathsf{MD}  = B_{3k} \mathcal{L} \bigl (c_{3k}^* c_{3k+3}^* - 3 e^{-2\gamma^2} c_{3k+1}^* c_{3k+2}^* \bigr) \mathcal{R} |\text{vacuum}\ra.
	\end{equation*}
	$\mathcal{L}$ is a sum over partitions of $\{0,\ldots,k-1\}$ and $\mathcal{R}$ is a sum over partitions of $\{k+2,\ldots, N-1\}$. 
	Eq.~\eqref{eq:bjzero} follows from the observation that $B_{3k}$ commutes or anticommutes with $\mathcal{L}$ and $\mathcal{R}$, and 
	\begin{equation*}
		B_{3k} \bigl (c_{3k}^* c_{3k+3}^* - 3 e^{-2\gamma^2} c_{3k+1}^* c_{3k+2}^* \bigr) |\text{vacuum}\ra  = 0.  \qedhere
	\end{equation*}
\end{proof}

\section{Renewal structure} \label{sec:renewal}

In this section we present the structure of the wave function that  is key to our proofs. 
This  structure  is already visible at the level of normalization constants, provided, however, 
we make a good choice for the multiplicative constant $\kappa_N$ in Eq.~\eqref{eq:wave-fct-period}. From here on we fix 
\begin{equation} \label{eq:kappan}
	\kappa_N := \frac {1}{\sqrt{N!}} \times \frac{1}{( 2\pi R \ell \sqrt{\pi} )^{N/2}}\times \exp\left( - \frac{1}{2} p^2 \gamma^2 \sum_{j=1}^N (j-1)^2 \right)
\end{equation}
(recall $\gamma = \ell /R$). 
With this choice, the normalization $C_N =||\Psi_N||^2_\mathcal{Z}$ becomes supermultiplicative~\cite[Sect. 3.2]{jls}.

In the remaining part of this article we choose length units such that the magnetic length is $\ell =1$.

\subsection{Block structure of the Vandermonde matrix}

Laughlin's wave function involves the $p$-th power of a $N\times N$  Vandermonde determinant in the 
variables $Z_j = \exp(\gamma z_j)$. For $1 \leq k \leq N$, the upper left $k\times k$ 
block of the Vandermonde matrix is itself a Vandermonde matrix. The  lower right 
$(N-k) \times (N-k)$ block is a Vandermonde matrix 
in which the row with variable $Z_j$ has been multiplied with $Z_j^k$: 
\begin{equation} \label{eq:block}
\begin{pmatrix}
	  1 & \cdots & Z_1^{k-1} & & & \\
		\vdots &  &\vdots & &\ddots & \\
		1 & \cdots & Z_k^{k-1}& & & \\
		 &         &          & Z_{k+1}^k&\cdots & Z_{k+1}^{N-1}\\
		 &  \ddots &          &  \vdots        &       & \vdots \\
		 &         &          & Z_N^k &\cdots & Z_N^{N-1}
		\end{pmatrix}^p.
\end{equation}
Thus the Vandermonde matrix has a simple block structure. 
Now, the power of the Vandermonde determinant is a polynomial and as such can be written as a sum 
of monomials. The block structure suggests that the expansion coefficients have some recursive structure. 
This is, indeed, true, as was shown in~\cite{fgil}.  Moreover, the recursive structure carries over to coefficients in 
Laughlin's wave function~\cite[Lemma 3]{jls}. 

Let us explain this in some more detail. Let $\psi_k(z)$ be the orthonormal basis functions from Eq.~\eqref{eq:basis}. 
We can write 
\begin{equation}  \label{eq:expansion}
	\Psi_N (z_1,\ldots,z_N) = \frac{1}{\sqrt{N!}} \sum_{m_1,\ldots,m_N } a_N(m_1,\ldots,m_N) \psi_{m_1}(z_1) \cdots \psi_{m_N} (z_N).
\end{equation}
The sum is over integer $m_1,\ldots,m_N$ running from $0$ to $pN-p$. 
Not all such $\vect{m}$'s contribute to the sum: if  $a_N(\vect{m}) \neq 0$ , then for all $k =1,\ldots,N$, 
		\begin{equation} \label{eq:dominance}
			m_1+\cdots + m_k \geq 0+p+\cdots +p (k-1)
		\end{equation} 
		with equality for $k=N$.  We call such $\vect{m}$'s \emph{admissible}. 
If $(m_1,\ldots,m_N)$  is admissible and in Eq.~\eqref{eq:dominance} there is equality  for some $k\leq N-1$ then 
		\begin{equation*}
			m_1,\ldots,m_k \leq pk-p,\quad pk \leq m_{k+1},\ldots,m_N 
		\end{equation*}
		and
		\begin{equation} \label{eq:product-rule}
			a_N(m_1,\ldots,m_N) = a_k(m_1,\ldots,m_k) a_{N-k}(m_{k+1} - pk,\ldots,m_N -pk). 
		\end{equation}
This product rule mirrors the block structure of the Vandermonde matrix~\eqref{eq:block}.
Not every $a_N(\vect{m})$ factorizes as in Eq.~\eqref{eq:product-rule}: some $\vect{m}$'s are \emph{irreducible}. 
Others, in contrast, may factor along more than one $k$.

\begin{example} 
	For $p=3$, the one- and  two-particle wave functions are  
	\begin{equation*}
		\Psi_1 = \psi_0, \qquad 
		\Psi_2 = \psi_0 \wedge \psi_3- 3 \exp(- 2\gamma^2) \psi_1 \wedge \psi_2.
	\end{equation*}
	$a_2(1,2)$ does not factorize, but $a_2(0,3)$ does: 
	\begin{equation*}
		a_2(0,3) = 1 = 1 \times 1 = a_1(0) \times a_1(3-3).
	\end{equation*}
\end{example}	

 \begin{remark}
	Product rules have been derived for other fractional Hall effect trial functions~\cite{bernevig}.
\end{remark}

\subsection{Ground state perturbation series}

It would be nice to have a physical interpretation of the product rule described in the previous section. 
Unfortunately, such an interpretation does not seem readily at hand. We hope, however, that the ground state perturbation 
series argument presented in this section provides some physical intuition. The reader should compare the product rule~\eqref{eq:product-rule} to a ``ground state concatenation'' equation, Eq.~\eqref{eq:concat} below.

In this section we consider only filling factor $1/p = 1/3$. 
We start with the Hamiltonian $H$, or its finite lattice version $H_\mathcal{L}$, defined in Sect.~\ref{sec:lattice}. 
For large $\gamma$, the dominant contribution to the Hamiltonian comes from 
 nearest and next-nearest neighbor repulsion
\begin{multline*}
 	H_\mathcal{L} = \sum_{k\in \mathcal{L}} \exp(- \gamma^2/2) \hat n_k \hat n_{k+1} 
			+  \sum_{k\in \mathcal{L}} 4 \exp(- 2\gamma^2) \hat n_k \hat n_{k+2} \\
		+ O\bigl( \exp( - 10 \gamma^2/4 ) \bigr).
\end{multline*}
(Thus we discard even more terms than in the definition of the monomer-dimer Hamiltonian.) 
An exact ground state of the nearest and next-nearest neighbor repulsion is, obviously, the so-called ``Tao-Thouless state''~\cite{tao-thouless}
\begin{equation*}
 	\Psi_N^\mathsf{TT} =  c_0^*\,  c_3^* \cdots c_{3N-6}^*\, c_{3N-3}^*  \, |\text{vacuum}\ra
\end{equation*}
which is why we will abbreviate the corresponding Hamiltonian $H^\mathsf{TT}_\mathcal{L}$.
For $\mathcal{L} =\{0,1,\ldots, pN-p\}$,  $\Psi_N^\mathsf{TT}$ is actually the unique 
ground state of $H^\mathsf{TT}_\mathcal{L}$. Thus starting 
from 
\begin{equation*}
 	H_\mathcal{L} \Psi_N = (H_\mathcal{L}^\mathsf{TT} + V_\mathcal{L}) \Psi_N =0,
	\quad H_\mathcal{L}^\mathsf{TT} \Psi_N^\mathsf{TT} =0,
\end{equation*}
we obtain 
\begin{equation} \label{eq:perturb}
 	\Psi_N = \Psi_N^\mathsf{TT} + \sum_{n=1}^\infty (-1)^{n} \Bigl( Q (\hat H_\mathcal{L}^\mathsf{TT} )^{-1} 
			Q V_\mathcal{L} Q \Bigr)^n \Psi_N^\mathsf{TT}.
\end{equation}
where $Q = Q_{N,\mathcal{L}}$ projects onto the orthogonal complement of $\Psi_N^\mathsf{TT}$ 
and $(\hat H_\mathcal{L}^\mathsf{TT})^{-1}$ is the inverse of $H_\mathcal{L}^\mathsf{TT}$ restricted 
to $\mathcal{R}(Q)$. For sufficiently large $\gamma$, 
the 
expansion~\eqref{eq:perturb} should be convergent. Eq.~\eqref{eq:perturb} is a perturbation series for the ground state, simplified 
because both the unperturbed and the perturbed ground state energy are equal to $0$. 

In the occupation number basis, 
 the Tao-Thouless Hamiltonian $H_\mathcal{L}^\mathsf{TT}$ 
and the projection $Q$ are diagonal,
 while $V_\mathcal{L}$ has non-diagonal,  ``hopping'',  terms, e.g., 
$c_0^* c_3^* c_2 c_1$. Thus Eq.~\eqref{eq:perturb} tells us that the ground state 
is a sum of 
the Tao-Thouless state  
and  orthogonal terms obtained from $\Psi_N^\mathsf{TT}$ by hopping a pair of particles,
	$(n_1,n_2) \to (n_1+ r, n_2 - r)$.

More precisely, let $\mathcal{B}_\mathcal{L}$ be the set of pairs $((k_1,k_2), (n_1,n_2))$   
of integers in $\mathcal{L}$ with $k_1+ k_2 = n_1 +n_2$, discarding pairs of the form 
$(k_1,k_2) = (n_1,n_2)= (k,k+1)$ or $=(k,k+2)$, or permutations thereof. Let 
\begin{equation*}
	A_\mathcal{L}(b):= - f(k_1-k_2) f(n_1-n_2)\,   Q (\hat H_\mathcal{L}^\mathsf{TT} )^{-1} 
			Q\,  c_{ k_1}^* c_{k_2}^* c_{n_2} c_{n_1}\, Q
 \end{equation*}
so that 
\begin{equation*}
	Q (\hat H_\mathcal{L}^\mathsf{TT} )^{-1} 
			Q V_\mathcal{L} Q = - \sum_{b \in \mathcal{B}_\mathcal{L}} A_\mathcal{L}(b).
\end{equation*}
 Thus $\Psi_N$ is a sum over bond-paths $(b_1,b_2,\ldots,b_n)$ of variable  
length $n$,   
\begin{equation}  \label{eq:bond-paths}
	\Psi_N 
	= \Psi_N^\mathsf{TT} + \sum_{n=1}^\infty \sum_{b_1,\ldots,b_n \in \mathcal{B}_\mathcal{L}} A_\mathcal{L}(b_1) \cdots A_\mathcal{L}(b_n)
		\Psi_N^\mathsf{TT}.  		 
\end{equation}
Now let $k \in \{1,\ldots, N-1\}$ and imagine drawing a vertical line at $x= pk - 1/2$ splitting the Tao-Thouless state 
in two. 
A path $(b_1,\ldots,b_n)$  \emph{crosses} the line if one of the bonds $b_i$ involves integers  both to the left and to the right of the line. 
If in Eq.~\eqref{eq:bond-paths} we discard crossing paths, 
the sum factorizes as 
\begin{equation} \label{eq:concat}
	\Psi_N^\mathsf{TT} + \sum_\text{non-crossing paths} A_\mathcal{L}(b_1) \cdots A_\mathcal{L}(b_n)
		\Psi_N^\mathsf{TT} 
		= \Psi_k  \wedge pk\text{-shifted} \Psi_{N-k},
\end{equation}
(or a symmetrized product instead of $\wedge$, depending on the parity of $p$). 
This is a
$k$-particle ground state to the left of the separation line, concatenated with 
an $(N-k)$-particle ground state to the right of the line, and is a version of the product rule~\eqref{eq:product-rule}. 

\begin{remark}
	The idea of treating the full Hamiltonian as a perturbed version of a thin cylinder 
	Hamiltonian is related to  ``one-dimensional'' approaches to the quantum Hall effect, see~\cite{bk} and the references therein. 
\end{remark}

\subsection{Renewal points and lattice partitions}

Here we explain how the product rule~\eqref{eq:product-rule} was exploited in~\cite{jls}. 
First we need some notation. We call the points along which $a_N(\vect{m})$ factorizes 
\emph{renewal points}. More precisely, let  $\vect{m}=(m_1,\ldots,m_N)$. If $\vect{m}$ 
is not increasing, let $\sigma$ be a permutation that rearranges the integers, $m_{\sigma(1)} \leq \cdots \leq m_{\sigma(N)}$. 
We call $r=pk$, $k=0,\ldots,N$, a renewal point of $\vect{m}$ if 
$\sum_1^k m_{\sigma(j)} =  \sum_1^k p(j-1)$. 
The end points $0$ and $pN$ are always renewal points.
The set of renewal points of $\vect{m}$ is denoted $\mathcal{R}(\vect{m})$.

Sometimes it is more convenient to work with partitions  rather than 
 renewal points. 
%
%
Thus let $\mathcal{P}_N$
be the set of partitions of the discrete volume 
$\{0,...,pN-1\}$ into discrete intervals or ``rods'' $X_1,...,X_D$ of the form 
$\{pj,\ldots, pj+pn-1\}$.  
The number $D$ of rods in the partition varies from $1$ to $N$, and the rods are always labeled 
from left to right, i.e., the elements of $X_1$ are smaller than those of $X_2$, etc. 
We use the short-hand 
\begin{equation*}
	\mathcal{X} = (X_1,\ldots,X_D) 
\end{equation*}
and let $\mathcal{R}(\mathcal{X})$ be the set of starting points $pj$ of $X_k$'s, with $pN$ added. 
We write $\mathcal{X}(\vect{m})$ for the partition associated with the renewal points of $\vect{m}$. 

It is useful to group together $\vect{m}$'s that give rise to the same partition. Thus we write 
\begin{equation} \label{eq:vect-decomp}
	\Psi_N (z_1,\ldots,z_N)= \sum_{\mathcal{X} \in \mathcal{P}_N} u_\mathcal{X}(z_1,\ldots,z_N)
\end{equation}
with
\begin{equation} \label{eq:fermionic-ux}
	u_\mathcal{X}(z_1,\ldots,z_N)
		:=  \frac{1}{\sqrt{N!}} \sum_{ \mathcal{X}(\vect{m}) = \mathcal{X}} 
			a_N(m_1,\ldots,m_N)  \psi_{m_1} (z_1) \cdots \psi_{m_N}(z_N). 
\end{equation}
As a consequence of the product rule, the function $u_\mathcal{X}(z_1,\ldots,z_N)$ factorizes into an (anti-)symmetrized product 
of functions $v_X$ associated with individual rods instead of the whole partition:  
if $\mathcal{X} = (X_1,\ldots,X_D)$ and $p$ is odd, 
\begin{equation} \label{eq:ux-vx}
	u_\mathcal{X}(z_1,\ldots,z_N) = v_{X_1}\wedge \cdots \wedge v_{X_D}(z_1,\ldots,z_N).
\end{equation}
Good orthogonality properties ensure that  
\begin{equation*} 	\label{eq:partfunct}
	||\Psi_N||_\Lambda^2 = \sum_{\mathcal{X} \in \mathcal{P}_N}||u_\mathcal{X}||_\Lambda^2  
			= \sum_{ (X_1,\ldots,X_D) \in \mathcal{P}_N} ||v_{X_1}||_\Lambda^2\times \cdots \times ||v_{X_D}||_\Lambda^2
\end{equation*}
whenever $\Lambda$ has the form $ [a,b] \times [0,2\pi R]$.  The family of functions $(v_X)$ has the following 
properties when $p$ is odd:   
\begin{itemize}
	\item $v_X$ is an antisymmetric function of $N(X)=n$ complex variables, 
					with $p n$  the cardinality of $X$, i.e.,  $ X = \{pj,\ldots, pj+pn-1\}$ for some $j$.  
	\item $v_X(z_1,\ldots,z_{N(X)})$ depends on $X$, $p$ and $\gamma$, but not on $N$. 
	\item $v_X$ is a linear combination of antisymmetrized products of the $\psi_k$'s with localized indices $k$: 
		\begin{equation} \label{eq:local}
			v_X \in \wedge^{N(X)}\mathrm{span}\{\psi_k\mid k \in X \}.
		\end{equation}
	\item 	Up to a $y$-dependent phase, 
		$	v_{X+p}(\{z_j\}) $ is the same as $v_X (\{z_j-p\gamma\})$. 
		Put differently, the shift of a rod is equivalent to the magnetic translation 
		of the associated function. (Here $p+X$ could be, e.g., $ 3+\{0,1,2\} = \{3,4,5\}$.)
\end{itemize}
Similar statements hold for even $p$, provided we replace antisymmetrized products with symmetrized products, e.g., in Eq.~\eqref{eq:ux-vx}.

For later purpose we note that the translational covariance allows us to define functions $v_X$ for intervals 
$X = \{pj,pj+1,\ldots, pj+pn - 1\}$ not necessarily contained in $\{0,1,\ldots, pN -1\}$, and functions $u_\mathcal{X}(z_1,\ldots,z_N)$ 
when $\mathcal{X}$ is a partition of some set $\{pj, pj+1, \ldots, pj+pN-1\}$, $j\in \Z$.

The decomposition of $\Psi_N$ as a sum over partitions was exploited in~\cite{jls} in order to deduce statements about 
correlation functions, or expectation values of observables. We will essentially follow this approach. First, however, 
let us explain how to go from a decomposition of the \emph{vector} $\Psi_N$ to a decomposition of the \emph{state} $|\Psi_N \ra \la \Psi_N|$.

\subsection{Quasi-state decomposition}

To each partition $\mathcal{X}= (X_1,\ldots,X_D)$ assign the weight 
 \begin{equation}  \label{eq:weight}
 	p_{N} (\mathcal{X}) 
 		= \frac{1}{||\Psi_N||^2} ||v_{X_1}||^2 \times \cdots \times 
 			||v_{X_D}||^2. 
 \end{equation}
Because of Eq.~\eqref{eq:partfunct} the weights sum up to $1$ and 
define a probability measure $\mathbb{P}_N$ on $\mathcal{P}_N$. 
It is natural to try to write the state $|\Psi_N\ra \la \Psi_N|$ as 
a weighted sum of states $\omega_\mathcal{X}$, associated each with a partition. 
Eq.~\eqref{eq:vect-decomp} immediately yields
\begin{equation*}
	|\Psi_N\ra \la \Psi_N| = \sum_{(\mathcal{Y},\mathcal{Z}) \in \mathcal{P}_N \times 
		\mathcal{P}_N} |u_\mathcal{Y} \ra \la u_\mathcal{Z}|. 
\end{equation*}
The state $|\Psi_N\ra\la \Psi_N|$ is, therefore, a sum over \emph{pairs} of partitions. 
In order to reduce this to a sum over single partitions, we group pairs according to their 
common renewal points.  For a given partition $\mathcal{X} \in \mathcal{P}_N$, let 
$\mathcal{M}_N(\mathcal{X})$ be the 
set of pairs $(\mathcal{Y},\mathcal{Z})$ such that 
$$ \mathcal{R}(\mathcal{Y}) \cap \mathcal{R}(\mathcal{Z}) = \mathcal{R}(\mathcal{X}) .$$
Thus in particular, $(\mathcal{X}, \mathcal{X}) \in \mathcal{M}_N(\mathcal{X})$. 
Let $\omega_\mathcal{X}$ be the operator
\begin{equation*} 
	\omega_\mathcal{X}:= \frac{1}{||u_\mathcal{X}||^2}
		\sum_{ (\mathcal{Y},\mathcal{Z}) \in \mathcal{M}_N(\mathcal{X}) }
	 \bigl| u_\mathcal{Y} \bigr \ra\, \bigl \la u_\mathcal{Z} \bigr|.		
\end{equation*}
By a slight abuse of language we use the same letter for the linear 
functional on operators, 
\begin{equation*}
	\omega_\mathcal{X}(a) =  \frac{1}{||u_\mathcal{X}||^2} 
		\sum_{ (\mathcal{Y},\mathcal{Z}) \in \mathcal{M}_N(\mathcal{X}) } 
			\la u_\mathcal{Z}, a\, u_\mathcal{Y} \ra. 
\end{equation*}
The sum contains, on top of the diagonal term $|u_\mathcal{X} \ra \la u_\mathcal{X} |$, 
off-diagonal terms where $\mathcal{Y} \neq \mathcal{Z}$. 
We can write, then, 
 \begin{equation} \label{eq:quasi-decomp}
 	\la \cdot \ra_N = \frac{1}{||\Psi_N||^2} |\Psi_N\ra \la \Psi_N| = \sum_{\mathcal{X} \in \mathcal{P}_N} 
 		p_{N} ( \mathcal{X}) \omega_\mathcal{X}. 
 \end{equation}
which is the decomposition we had looked for.

\begin{example}
  Consider the two-particle wave function at filling factor $1/3$. Thus $N=2$ and $p=3$. 
    The discrete volume $\{0,1,\ldots,5\}$ has two admissible partitions:
  $\mathcal{P}_2 = \{\mathcal{V},\mathcal{W}\}$ with 
  \begin{equation*}
      \mathcal{V}=( \{0,1,2\}, \{3,4,5\}),\quad 
    \mathcal{W}=(\{0,1,\ldots,5\}).
  \end{equation*}
  We have $\Psi_2 = u_\mathcal{V} + u_\mathcal{W}$ with 
  \begin{equation*}
    u_\mathcal{V} = \psi_0 \wedge \psi_3,\quad u_\mathcal{W} = - 3 \exp(-2 \gamma^2) \psi_1 \wedge \psi_2.
  \end{equation*}
  The corresponding probability distribution is 
  \begin{equation*}
    p_2(\mathcal{V}) = \frac{1}{1+9 \exp(-4 \gamma^2)},\quad 
    p_2(\mathcal{W}) = \frac{9 \exp(-4\gamma^2)}{1+9 \exp(-4 \gamma^2)}.
  \end{equation*}
  The state decomposes as $\la \cdot \ra_2 = p_2(\mathcal{V}) \omega_\mathcal{V} + p_2(\mathcal{W}) \omega_\mathcal{W}$
  with 
  \begin{align}
    \omega_\mathcal{V} & = |03\ra \la 03|, \notag \\ 
    \omega_\mathcal{W} & = |12 \ra \la 12|  - \frac{1}{3} \exp(2 \gamma^2) 
      \Bigl( |12 \ra \la 03| + |03 \ra \la 12| \Bigr) \label{eq:deco}
  \end{align}
  where we have used the short-hand $|03\ra = |\psi_0 \wedge \psi_3\ra$. 
\end{example}

An important remark, immediate from Eq.~\eqref{eq:deco}, is that 
the matrices $\omega_\mathcal{X}$ are in general \emph{not} density matrices, 
and in general we do \emph{not} have $||\omega_\mathcal{X}(a)|| \leq ||a||$. 
The picture changes, however, if we evaluate only diagonal operators: let 
$a$ be a product of occupation numbers $\hat n_k$, or any other operator 
which is diagonal in the basis of Slater determinants $\psi_{m_1}\wedge \cdots \wedge 
\psi_{m_N}$. The functions $u_\mathcal{X}$ have a strong orthogonality 
property which ensures that
\begin{equation*}
	\Bigl(a\ \text{diagonal and } \mathcal{Y} \neq \mathcal{Z}\Bigr)\ 
		\Rightarrow\ \la u_\mathcal{Y}, a\, u_\mathcal{Z} \ra = 0.
\end{equation*}
As a consequence, if $a$ is diagonal, 
\begin{equation*}
	\omega_\mathcal{X}(a):= \mathrm{tr}\, (\omega_\mathcal{X} a) 
		= \la u_\mathcal{X}, a\, u_\mathcal{X} \ra / ||u_\mathcal{X}||^2.
\end{equation*}
Thus as far as diagonal operators are concerned, $\omega_\mathcal{X}$ \emph{is} 
a state. In the language of~\cite{an}, $\omega_\mathcal{X}$ is 
a \emph{quasi-state}: a linear map on a $C^*$ algebra whose restriction to a commutative 
sub-algebra is a state, and  Eq.~\eqref{eq:quasi-decomp} is a \emph{quasi-state decomposition}. 

The quasi-state decomposition allows us to shift our focus from the complicated 
quantum mechanical state to the probability measure $\mathbb{P}_N$, which is a much simpler 
object. This reduction is useful because the $\omega_\mathcal{X}$ themselves are reasonably simple, 
as ensured by the following properties:
\begin{itemize}
	\item \underline{Clustering}: let $a, b$ be 
	local observables. ``Local'' refers to the lattice picture: $a$ and $b$ are, for example, Wick-ordered monomials, 
	and if $a = c_1^* c_3$, we call the set $\{1,3\} = \supp a$ the support of $a$. 
	Let $\mathcal{X} \in \mathcal{P}_N$ be a partition whose 
	renewal points separate $a$ and $b$, i.e., there 
	is a $r \in \Z$ such that 
	\begin{equation*}
		r \in \mathcal{R}(\mathcal{X}),\ \supp a \subset \{\ldots, r-2,r-1\},
		\ \supp b\subset \{r,r+1,\ldots \}.
	\end{equation*}
	Then
	\begin{equation*}
		\omega_\mathcal{X}(ab) = \omega_\mathcal{X}(a)\, \omega_\mathcal{X}(b).
	\end{equation*}
	
	\item \underline{Translational covariance}: 
	Shifting a partition is equivalent to shifting an observable: 
	 for every observable $a \in \mathcal{A}$, 
	\begin{equation*}
		\omega_{p+\mathcal{X}} (a) = \omega_\mathcal{X} (\tau_x^{p}(a)).
	\end{equation*}
	Here for example, $\tau_x^3(c_2^* c_6^*) =c_5^* c_9^*$. More generally, $\tau_x$ is 
	the automorphism on the observable algebra induced by the magnetic translation $t(\ell^2/R\, \vect{e}_x)$  in the one-particle 
	Hilbert space. Note that the shift transforms a partition of $\{0,\ldots, pN-p\}$ into a partition of $\{p,\ldots, pN\}$. 

	\item \underline{Locality}: let $a$ be a  
	local observable. Let $N\in \N$ and $\mathcal{X}, \mathcal{X'} \in \mathcal{P}_N$ 
	be two partitions. Suppose that 
	$\mathcal{X}$ and $\mathcal{X}'$ coincide on some interval containing $\supp a$, i.e., they  have common renewal points $r,s$ such that 
	that $\supp a \subset \{r,\ldots, s-1\}$ 
	and 
	\begin{equation*}
		\mathcal{R}(\mathcal{X}) \cap \{r,\ldots,s\} 
			= \mathcal{R}(\mathcal{X'}) \cap \{r,\ldots,s\}.
	\end{equation*}
	Then
	\begin{equation*}
		\omega_\mathcal{X} (a) = \omega_\mathcal{X'}(a).
	\end{equation*}
\end{itemize}
These properties are, again, a consequence of the product rule Eq.~\eqref{eq:product-rule}.

For later purpose let $\mathcal{P}_\infty$ be the set of partitions of $\Z$
into sets of the form $\{pj,pj+1,\ldots, pj+pn-1\}$,  $n\in \N$, $j\in \Z$. Thus we explicitly forbid partitions that contain an infinite component, and  every lattice site has 
 a renewal 
point to its left and a renewal point to its right. For $\mathcal{X} \in \mathcal{P}_\infty$, and $a$ a local observable, 
 let $r$ and $s$ be two renewal points of $\mathcal{X}$ enclosing the support of $a$. We can restrict $\mathcal{X}$ 
to a partition $\mathcal{X}_{r,s}$ of $\{r,\ldots, s-1\}$ and set 
\begin{equation*}
	\omega_\mathcal{X}(a):= \omega_{\mathcal{X}_{r,s}} (a). 
\end{equation*}
Locality ensures that this definition does not depend on the precise choice of renewal points $r,s \in \mathcal{R}(\mathcal{X})$. 

In this way we obtain a family of linear maps $ a \mapsto \omega_\mathcal{X}(a)$, indexed by partitions 
of $\Z$,  $\mathcal{X}\in \mathcal{P}_\infty$, and defined on local observables $a$.
 These maps inherit the clustering, locality and translational covariance from their finite volume counterparts. 
They will allow us to write the infinite volume state as an integral, 
\begin{equation*}
	\frac{\la \Psi_N,a \Psi_N\ra}{||\Psi_N||^2} \to  \int_{\mathcal{P}_\infty} \omega_\mathcal{X}(a) \, \dd\P(\mathcal{X})
\end{equation*}
in terms of a suitable probability measure $\P$ on partitions of $\Z$. 

\section{Peierls type argument  and consequences} \label{sec:peierls}

As explained in the introduction, the algebraic structure described in the previous subsections 
allows for a considerable simplification of the problem of thermodynamic limits: instead of looking at a full quantum-mechanical state, we can look at a simpler probability measure $\P_N$. 
A further simplification is that $\P_N$ can be shown to have a stationary limit $\P$ if a certain condition on the asymptotics of normalization constants is satisfied, see~\cite{jls} and Sect.~\ref{sec:finite-mean}.

The aim of this section is to prove that this condition is indeed satisfied. This is shown by viewing the normalization as a partition function for a Coulomb gas. In Sect.~\ref{sec:pef}, we adapt the methods of~\cite{ajj} to prove that with positive probability, the system splits into neutral subsystems with finite interaction, and deduce the required statements on asymptotics 
(Sect.~\ref{sec:normalization}). 
 
We also derive an auxiliary bounds on correlation functions, needed for bosons (even $p$), in Sect.~\ref{sec:bosons}.

Throughout this section we fix $\Lambda = \mathcal{Z}$, i.e., all integrations are on the infinite cylinder, 
and $||\cdot ||=||\cdot||_\mathcal{Z}$ refers to the $L^2$ norm on the infinite cylinder.

\subsection{Normalization on an infinite cylinder} \label{sec:normalization}
Let 
$$ C_N : = ||\Psi_N||^2 = \int_{\R^N} \dd ^N \boldsymbol{x} \int_{[0,2\pi R]^N} \dd^N \boldsymbol{y}\thinspace |\Psi_N(z_1,...,z_N)|^2.$$ 
In~\cite{jls}, it was shown that the limits 
\begin{equation} \label{eq:norm-known}
	\lim_{N\to \infty} \frac{1}{N} \ln C_N =: - \log r(p,\gamma),\quad 
	\lim_{N\to \infty} C_N r(p,\gamma)^N =: \Bigl( \mu(p,\gamma)\Bigr)^{-1} 
\end{equation}
exist and are finite, with $1/\mu(p,\gamma) =0$  not yet excluded:
$$ 0 < r(p,\gamma) \leq 1,\quad  0 <\mu(p,\gamma) \leq \infty. $$ 
As was shown in~\cite{jls}, on thin cylinders ($\gamma$ large), 
$\mu(p,\gamma)$ is certainly finite. However, what happens on thick cylinders was left open. Here we answer this question:  
 $\mu(p,\gamma)$ is finite, no matter how large the cylinder radius is.

\begin{theorem} \label{thm:finite-mean}
	For all $p$ and $\gamma$, the quantity $\mu(p,\gamma)$ is finite. 
\end{theorem}

The proof of the theorem uses a form of submultiplicativity, 
which complements the supermultiplicativity $ C_{N+M} \geq C_N C_M$:

\begin{prop}[Submultiplicativity] \label{lemma:sub-mult}
	For all $N,M$ and a suitable constant $c(p,\gamma)>0$, 
	\begin{equation} \label{eq:sub-mult}
	  C_{N+M} \leq c(p,\gamma) C_N C_M.
	\end{equation}
\end{prop}
The proof of the submultiplicativity is deferred to the next section.

\begin{proof}[Proof of Theorem \ref{thm:finite-mean}]
	All we need to show is that for all $N$ and some $\eps>0$, 
	$ C_N r(p,\gamma)^N \geq \eps$.
	But this follows from Eq.~\eqref{eq:sub-mult} and~\eqref{eq:norm-known}, by standard arguments: 
	 repeated application of Eq.~\eqref{eq:sub-mult} yields 
	$$ C_{mN} \leq c(p,\gamma)^{m-1} C_N^m.$$ 
	We take the logarithm,  divide by $mN$, and let $m\to \infty$. This gives 
	$$ - \log r(p,\gamma) \leq \frac{1}{N} \log c(p,\gamma) + \frac{1}{N} \log C_N .$$ 
	It follows that  $C_N r(p,\gamma)^N \geq 1/c(p,\gamma)$ for all $N$ and 
	$ \mu(p,\gamma) \leq c(p,\gamma) <\infty.$ 
\end{proof}

\subsection{Particle excess function} \label{sec:pef}

Recall from Eq.~\eqref{eq:wave-fct-period} that $|\Psi_N(z_1,\ldots,z_N)|^2$ favors 
particle abscissas $x_k = (k-1) p \gamma$, $k =1,\ldots, N$. Think of the infinite 
cylinder as a collection of $N$ annuli of width $p  \gamma$, centered around those optimal abscissas, 
and a left and right tail. Let $\bar x$ be at the boundary between two annuli, i.e., 
\begin{equation*}
	\bar x = (k- 1/2) p\gamma 
\end{equation*}
for some $k = 1, \ldots,N-1$. The state should prefer configurations with $k$ particles to the left of $\bar x$, and $N-k$ particles to 
the right of $\bar x$. Deviations from this optimum configuration are measured by the particle excess function 
\begin{equation*}
	K(\bar x; z_1,\ldots,z_N) := \# \{ j \mid x_j \leq \bar x\} - k = N-k - \# \{j \mid x_j> \bar x\}.
\end{equation*}
\begin{theorem} \label{thm:zero-imbalance}
	Let $P_N$ be the probability measure on $\mathcal{Z}^N$ with density proportional to $|\Psi_N(z_1,\ldots,z_N)|^2$. 
	Let $\bar x$ and $K(\bar x; \vect{z})$ as above. Then, for some constant $c(p,\gamma)< \infty$ independent of $N$ and $\bar x$ 
	\begin{equation*}
		P_N\bigl( K(\bar x; z_1,\ldots,z_N) = 0 \bigr) \geq 1/c(p,\gamma) >0.
	\end{equation*}
 \end{theorem}
Prop.~\ref{lemma:sub-mult} will follow from a similar lower bound on the probability that not only the particle excess vanishes, 
but moreover particles do not accumulate too close to $\bar x$, see the definition of regular configurations $\Omega_\mathrm{reg}$ below. 
The next theorem will be useful for the bosonic wave function.
\begin{theorem} \label{thm:imbalance}
	With the notation of Theorem~\ref{thm:zero-imbalance}:
  For suitable 
  constants $C,c>0$ and all $n \in \N_0$
  \begin{equation*}
   P_N\bigl( |K(\bar x; z_1,\ldots,z_N )|  \geq n\bigr) \leq C \exp( - c n ^3).
  \end{equation*}
	$C$ and $c$ do not depend on $N$ or $\bar x$. 
\end{theorem}
The proofs follow ideas from~\cite{ajj}.  They are best understood in the light of 
the plasma analogy: we think of $|\Psi_N|^2$ as the Boltzmann weight for a classical Coulomb system. 
Note that Theorems~\ref{thm:zero-imbalance} and~\ref{thm:imbalance} only involve the modulus of the wave function and
they do not refer to the full quantum mechanical setting.

\begin{proof}[Proof of Prop.~\ref{lemma:sub-mult}]
	Let $N, M\in \N$. 
	 Taking into account proportionality constants,  
	 Eq.~\eqref{eq:wave-fct-period} becomes
	\begin{equation*}
		|\Psi_{N+M} (z_1,\ldots,z_{N+M})|^2 = \frac{(2 \pi \sqrt{\pi}/\gamma)^{-(N+M)} }{(N+M)!}\, 
			\exp(-\, U(z_1,\ldots,z_{N+M}))
	\end{equation*}
	The ``energy'' $U(\boldsymbol{z})$ is, up to a multiplicative and an additive constant, the energy 
	of some classical quasi 1D Coulomb system. It is the sum of a one-dimensional energy and 
	a correction, 
	\begin{equation*}
		U(z_1,\ldots, z_N) : = U^1(x_1,\ldots,x_N) + \sum_{1 \leq j <k \leq N} V_2(z_j-z_k) 
	\end{equation*}
	where, for $x_1 \leq \cdots \leq x_{N+M}$, 
	\begin{align*}
		U^1(x_1,\ldots, x_{N+M}) & : =  \sum_{j=1}^{N+M} \Bigl( x_j - (j-1)p\gamma \Bigr)^2 \\
		V_2(z) & = - 2p \ln \Bigl| 1 - \exp\Bigl( - \gamma(|x| +i y) \Bigr) \Bigl|.
	\end{align*}
	Let 
	$$ \Omega:= \{ \boldsymbol{z}\in \C^{N+M} \mid  x_1 \leq \cdots \leq x_{N+M} , 
	\ 0 \leq y_j \leq 2 \pi R\}$$ 
	be the configuration space for $N+M$ particles on the infinite cylinder, labeled from left to right. 
	We imagine the cylinder split into two half-infinite cylinders, separated by $\bar x := (N-1/2) p \gamma$. 
	Let $\Omega_\mathrm{reg} \subset\Omega$ be the set of configurations that are regular in the following sense:
	\begin{itemize}
		\item [(i)] Each particle is in the ``correct'' half-cylinder, 
		 $$ x_1 \leq \cdots \leq x_N \leq (N- \frac{1}{2}) p \gamma \leq x_{N+1} \leq \cdots \leq x_{N+M} .$$
		\item [(ii)] Particles do not accumulate at the boundary: for all $j=1,\ldots,N+M$,
		 \begin{equation*}
			\Bigl|x_j - \Bigl(N- \frac{1}{2} \Bigr) p \gamma \Bigr| \geq \Bigl|j - N - \frac{1}{2} \Bigr|
			\, \frac{p\gamma}{2}. 
		  \end{equation*}
	\end{itemize}
	In $\Omega_\mathrm{reg}$, the system's energy $U(z_1,\ldots,z_N)$ is the sum of energies associated 
	with the subsystems of $N$ and $M$ particles, plus an interaction term that is lower bounded by
	\begin{align*}
		V_\mathrm{LR}(\boldsymbol{z}) & := 2p \sum_{j=1}^{N} \sum_{k=N+1}^{N+M} V_2(z_j-z_k) \\
			& \geq - 2p \sum_{j=1}^{N} \sum_{k=N+1}^{N+M} \ln [1+ \exp(- \gamma |x_{j}-x_{k}|) ] \\
			& \geq  - 2 p \Bigl( 1- \exp(- p\gamma^2/2) \Bigr)^{-2}, 
	\end{align*}
	for all $N$ and $M$. 
	As a consequence,
	\begin{equation*}
		(2\pi \sqrt{\pi}/\gamma)^{-(N+M)} \int_{\Omega_\mathrm{reg}} \exp(-\, U) 
			\leq \const C_N\,C_M, 
	\end{equation*}
	for some suitable $N,M$ independent constant.  Eq.~\eqref{eq:sub-mult} follows immediately once we know that 
	regular configurations have positive probability, i.e., 
	\begin{equation} \label{eq:pos-prob}
		\int_{\Omega}\exp(- U ) 
		\leq  \const \int_{\Omega_{reg}}\exp(- U ) , 
	\end{equation}
	uniformly in $N$ and $M$. 

	So it remains to show Eq.~\eqref{eq:pos-prob}. To this aim we use a Peierls-type argument: with 
	each irregular configuration we associate a regular configuration that has a smaller energy. This is 
	done by shifting particles that are in the wrong half-cylinder, or too close to the boundary, 
	closer to their optimum positions $(j-1)p\gamma$. Some technicalities arise because the implementing map 
	$T: \Omega\to \Omega_\mathrm{reg}$ 
	is not one to one and has non-trivial Jacobian, resulting in an entropy loss that has to be carefully 
	evaluated. 

	We start with the definition of the map $T:\Omega \to \Omega_\mathrm{reg}$. Let $\boldsymbol{z} \in \Omega$. 
	Let $j \leq N$. We shift the $j$-th particle as follows: 
	\begin{itemize}
		\item if $x_j > \bar x$ (irregular particle, wrong half-cylinder), then 
			$$ x'_j:= (j-1) p \gamma - \exp(-|x_j- \bar x|).$$ 
		\item if $\bar x - (N + 1/2 -j) p\gamma/2 <x_j \leq \bar x $ (irregular particle, 
		too close to the boundary), 
			then 
			$$ x'_j:= (j-1) p \gamma + \frac{|x_j - \bar x|}{N+1/2-j}. $$
		\item if $x_j \leq \bar x - (N + 1/2 -j) p\gamma/2$ (regular particle), then $x'_j:=x_j$.
	\end{itemize}
	Particles belonging to the right half-cylinder, $j \geq N+1$, are shifted in an 
	analogous way. To complete the definition of $T$, let $y'_j := y_j$ and 
	$$ T(\boldsymbol{z}) = (z'_{\sigma(1)},\ldots,z'_{\sigma(N+M)})$$ 
	with $\sigma$ a permutation that reorders the shifted particles from left to right.
	One can check that $T$ maps $\Omega$ into $\Omega_\mathrm{reg}$.

	It is convenient to group together 
	configurations that have the same irregular particles. Thus, for $J \subset \{1,\ldots, N+M\}$ 
	we let $\Omega_J$ be the configurations that have $x'_j \neq x_j$, 
	if and only if $j \in J$. This gives a partition of $\Omega$. 
 	$$\Omega = \bigcup_{J \subset \{1,\ldots, N+M\}} \Omega_J, \quad \Omega_\emptyset = \Omega_\mathrm{reg}.$$ 
	We claim that for suitable constants $c_1,c_2$, independent of 
	$N$, $M$ and $J$,  
	\begin{equation} \label{eq:oj}
		\int_{\Omega_J} \exp(-\, U)  \leq c_1\exp\Bigl( - c_2 \sum_{j \in J} (j-N)^2 \Bigr)
			\int_{\Omega_\mathrm{reg}} \exp(-\,U). 
	\end{equation}
	Eq.~\eqref{eq:pos-prob} follows from Eq.~\eqref{eq:oj} by a summation over subsets $J \subset \{1,\ldots,N+M\}$, 
noting that 
  \begin{equation*}
    \sum_{J \subset \{1,\ldots,N+M\}} \exp\bigl( - c_2 (j - N)^2 \bigr) \leq \left(\prod_{k=0}^\infty\bigl(1+ \exp( - c_2 k ^2)\bigr) 
      \right)^2 <\infty. 
  \end{equation*}
	We are finally left with the proof of Eq.~\eqref{eq:oj}. It  is shown with the help of a change of variables. We refer the reader to~\cite{ajj} for the 
	details and content ourselves here with some hints about the necessary estimates:

	 \underline{Entropy}: One has to give an upper bound of  the maximum number of preimages $|T^{-1}(\vect{z'})|$ of points in $\vect{z'} \in T(\Omega_J)$, and 
			a lower bound for  Jacobian of $|\dd T(\vect{z}) /\dd \vect{z}| $ for $\vect{z} \in \Omega_J$; note that $T$ is almost everywhere differentiable.
			The bounds depend on the set of irregular particle labels and are bad when there are a lot of irregular particles, but this is 
			compensated  by a gain in one-dimensional energy $\Delta U^1$. 

	\underline{Energy}: The simplest to estimate is the decrease in 1D energy, 
			\begin{equation*}
				\Delta U^1:= U^1(\boldsymbol{z}) - U^1(T(\boldsymbol{z}))
					\geq k_1 \sum_{j\in J} (j-N)^2 + k_1 \sum_{j:\ \text{wrong} \atop \text{half-cyl.}}
						(x_j - \bar x)^2
			\end{equation*}
			for some suitable constant $k_1$. 

			For the $V_2$-interactions, note that only the interactions 
			affecting irregular particles are changed. Now, instead of estimating  the change in
			 $V_2$-interaction directly, we adopt a three step procedure. First, we drop the affected  $V_2$-interaction
			 altogether, using 
			\begin{equation*}
				\sum_{j \in J} \sum_{k=1}^{N+M} V_2(z_j - z_k) \geq - k_2 |J|^2
			\end{equation*}
			for some constant $k_2$. 
			Next, we shift $x_j \mapsto x'_j$. Finally, we reinsert the dropped interactions
			with Jensen's equality: for fixed 
			$x\neq 0$, $V_2$ is a harmonic function of $y$, hence $\int V_2(x+iy) \dd y=0$ and 
			\begin{equation*}
				\int 1\, \dd \boldsymbol{y'}  \leq \int \exp\Bigl( - \sum_{j \in J} \sum_{k=1}^{N+M}
				 V_2(z'_j - z'_k) \Bigr)  \dd \boldsymbol{y'}.
			\end{equation*}
			All $y$-integrations go from $0$ to $ 2\pi R$. 
			At the end, we obtain an estimate on $y$-averaged Boltzmann weights,
			\begin{equation*} 
				\int \exp(-\,U(\boldsymbol{z}))\dd \boldsymbol{y} 
				\leq \exp\Bigl(k_2 |J|^2 - \Delta U^1 \Bigr)
					  \int \exp(-\, U(T(\boldsymbol{z}))) \dd \boldsymbol{y'},
			\end{equation*}
			with $\int \dd \boldsymbol{y'}$ the integration 
			over $y$-coordinates of $T(\boldsymbol{z})$. 
Eq.~\eqref{eq:oj} is obtained by combining this last estimate with the entropy estimates. 
\end{proof}

\begin{proof}[Proof of Theorem~\ref{thm:zero-imbalance}]
	In the notation of the proof of Prop.~\ref{lemma:sub-mult}, regular configurations have 
	zero particle imbalance and, therefore, Eq.~\eqref{eq:pos-prob} gives
	\begin{equation*}
		P_{N+M}\bigl( K(\bar x; z_1,\ldots,z_{N+M}) \bigr) \geq P_{N+M}(\Omega_\mathrm{reg}) \geq 1/\const >0, 
	\end{equation*}
	uniformly in $N$ and $M$. 
\end{proof}

\begin{proof}[Proof of Theorem~\ref{thm:imbalance}]
	Again, it is enough to have a closer look at the proof of Prop.~\ref{lemma:sub-mult}. We remark that if the particle excess is 
	negative, there are too many particles in the right half-cylinder: 
	if $K(\bar x; z_1,\ldots,z_{M+N}) \leq - k$ for $k \in \N$, then 
	 $x_{N} \geq \cdots \geq x_{N-k+1} \geq \bar x$. 
	Thus the particles with labels $N-k+1,\ldots, N$ are irregular. Since 
	\begin{equation*}
	  \sum_{ J \subset \{1,\ldots, N+M\} \atop \{N-k+1,\ldots, N\} \subset J} 
		\exp\bigl( - c_2 \sum_{j \in J} (j - N)^2 \bigr) 
	    \leq \exp\bigl( - c_2 \sum_{s=0}^{k-1} s^2 \bigr) \left(\prod_{j=0}^\infty 
			\bigl( 1+ \exp( - c_2 j ^2) \bigr) \right)^2
	\end{equation*}
	we find that the weight of configurations with particle excess $ \leq -k $ is of order 
	$\exp( - \const k ^3)$, uniformly in $N$ and $M$. A similar reasoning can be applied to 
	positive particle excess. This proves Theorem~\ref{thm:imbalance}.  
\end{proof}

\subsection{Uniform moment bounds for lattice bosons} \label{sec:bosons}

Before we investigate the thermodynamic limit, we give bounds on 
correlation functions in terms of the 1D lattice system. This is an auxiliary 
result needed for bosons ($p$ even) only. For fermions, the bounds given below are trivial 
since there can be at most one fermion per lattice site. 

Recall that 
$c_k^*, c_k $, $k \in \Z$,  are the creation and annihilation operators for 
the orbital $\psi_k(z)$.
Also, we denote $\la a \ra_N:= \la \Psi_N, a \Psi_N\ra / C_N$, with integration on the infinite cylinder. 

\begin{prop} \label{prop:corr-bounds} 
  Let  $n \in \N$. Then for some suitable 
  constant $D>0$, all $n,N \in \N$ and all
  $k_1,\ldots, k_n, m_1,\ldots, m_n \in \Z$, 
 \begin{align}
	\bigl| \la \hat n_{k_1} \cdots \hat n_{k_n} \ra _N \bigr| &\leq D^n \label{eq:occup-bounds}\\
    \bigl| \la c_{k_1}^* \cdots c_{k_n}^* c_{m_n} \ldots c_{m_1} \ra_N \bigr|  &  \leq D^n. \label{eq:factor-bounds}
 \end{align}
\end{prop}

\begin{proof}
We start with the proof of the first inequality in the case $k_1 = \cdots = k_n =k$. Thus we seek 
to estimate $\la \hat n_k^n\ra_N$. 
 Let  $A:=[a,b)\times[0,2\pi R]$ 
 be  the annulus of width $\gamma$, centered at $x=k \gamma$, i.e., $a,b=k\gamma \pm \gamma/2$. For 
	a given configuration $(z_1,\ldots,z_N)$, 
 the number of particles  in $A$ is 
 \begin{equation} \label{eq:nak}
    N_{A}(z_1,\ldots,z_N)
      = 1 + K(b;z_1,\ldots,z_N) - K(a; z_1,\ldots,z_N). 
  \end{equation}
  It follows that $N_A$ has moments of all orders, which by Theorem~\ref{thm:imbalance} are uniformly bounded, 
	$E_N(N_A^n) \leq w^n$ for all $n,N,A$ and some suitable $w>0$. Here $E_N$ refers to the measure $P_N$ on $\mathcal{Z}^N$ 
	with density $\propto |\Psi_N(z_1,\ldots,z_N)|^2$, and we use that $\sum_k k^n \exp( - k ^3)$ grows at most exponentially in $n$.  
In order to translate bounds from the continuum picture to the lattice picture, we use the formula 
for the $n$-point correlation function at $z_j = z'_j$: 
 \begin{multline} \label{eq:n-point}
    \rho_n^N(z_1,\ldots,z_n;z_1,\ldots,z_n) =  \sum \la c_{k_1}^*\cdots c_{k_n}^* c_{m_n} \cdots c_{m_1} \ra_N \\
	\overline{\psi_{k_1}(z_1)}\cdots \overline{\psi_{k_n}(z_n)}\, \psi_{m_1}(z_1)\cdots \psi_{m_n}(z_n).  
  \end{multline}
   The sum is over the integers $k_j, m_j \in \Z$. 
   When we integrate over the $y$-coordinates, 
  the contributions with $k_j \neq m_j$ vanish and the non-vanishing contributions are positive. 
Hence, integrating over $A^n$, we obtain 
\begin{equation*}
	E_N\Bigl( N_A(N_A - 1) \cdots (N_A - n+1) \Bigr) \geq u(\gamma)^n  \bigl \la (c_k^*)^n (c_k)^n \bigr \ra_N.
\end{equation*}
The constant $u(\gamma)<1$ comes from integrating the Gaussian $|\psi_{k}(z)|^2$ over the annulus $A$; it does not depend on $N$ or $k$. 
For the left-hand side, we have used that the diagonal $n$-point correlation is a factorial moment density (see \cite[Sect. 5.4]{dvj} for an explanation of this notion). 
The right-hand side is  a factorial moment too, since
\begin{equation*} 
	(c_k^*)^n (c_k)^n = \hat n_k (\hat n_k -1 ) \cdots (\hat n_k - n+1)=: {\hat n_k}^{[n]}.
\end{equation*}
Now, moments are linear combinations of factorial moments 
with positive coefficients $\Delta_{j,n} \geq 0$, the Stirling numbers of the second kind~\cite[Sect. 5.2]{dvj}. 
We deduce 
\begin{equation*}
	E_N\bigl( N_A^n) = \sum_{j=1}^n \Delta_{j,n} E_N\bigl( N_A^{[j]} \bigr) 
		 \geq \sum_{j=1}^n \Delta_{j,n}  u(\gamma) ^j \bigl \la \hat n_k^{[j]} \bigr \ra_N 
		 \geq u(\gamma)^n  \bigl \la \hat n_k^{n} \bigr \ra_N. 
\end{equation*}
 It follows that 
\begin{equation*}
	\bigl \la \hat n_k^n \bigr \ra_N \leq (w / u(\gamma)) ^n =: D^n
\end{equation*}
for all $n$. Repeated use of the Cauchy-Schwarz inequality $|\la a^*b \ra| \leq  \la \la a^*a\ra^{1/2} \la  b^* b\ra^{1/2}$ yields
Eq.~\eqref{eq:occup-bounds}. Since factorial moments are smaller than moments,  
it follows that Eq.~\eqref{eq:factor-bounds} holds when $k_1=m_1,\ldots, k_n =m_n$, from which the general case $k_j \neq m_j$ is 
deduced, again, with the help of  Cauchy-Schwarz. 
\end{proof}

As an application we prove Prop.~\ref{prop:off-diagonal}. 
\begin{proof}[Proof of Prop.~\ref{prop:off-diagonal} when $\Lambda = \mathcal{Z}$]
	We know that $\la c_k^* c_n\ra_N \neq 0$ whenever $k\neq n$, see~\cite{jls}, thus 
	\begin{equation*}
		\rho_1^N(z;z') = \sum_{k=0}^{pN-p} \la \hat n_k \ra_N  \psi_k(z) \overline{\psi_k(z')}.
	\end{equation*}
	For every integer $k$, we can factorize  
	\begin{align*}
		\bigr| \psi_{k}(z) \overline{\psi_{k}(z')}\bigl| &  \propto \exp \Bigl( - \frac{1}{2} \bigl( (x - k \gamma)^2 +(x'- k \gamma)^2 \bigr) \Bigr) \\
				& = \exp\Bigl( - \frac{1}{4} \bigl( x-x')^2 \Bigr) \exp\Bigl( - \frac{1}{4} \bigl( x+x' - 2 k\gamma\bigr)^2\Bigr). 
	\end{align*}
	By Prop.~\ref{prop:corr-bounds}, the occupation numbers are uniformly bounded. The proof is concluded by noting that 
	$\sum_{s\in \Z}  \exp( - (x - s\gamma )^2 /4)$ is bounded too, uniformly in $x\in \R$. 
\end{proof}

\subsection{Stationary renewal process} \label{sec:finite-mean}

Recall the weights $p_N(\mathcal{X})$ and the corresponding measure $\P_N$ on $\mathcal{P}_N$, 
see Eq.~\eqref{eq:weight}. When we integrate on the infinite cylinder, the translational covariance of 
the functions $v_X(z_1,\ldots,z_{N(X)})$ gives 
Eq.~\eqref{eq:partfunct} the form 
\begin{equation*} \label{eq:cnan}
	C_N= ||\Psi_N||^2 
			=\sum_{n_1+\cdots +n_D=N}\alpha_{n_1}\times \cdots \times \alpha_{n_D}.
\end{equation*}
As shown in~\cite{jls}, the positive numbers $\alpha_n$ relate to the quantities $r(p,\gamma)$ and $\mu = \mu(p,\gamma)$ 
from Section~\ref{sec:normalization} through 
\begin{equation*}
	\sum_{n=1}^\infty \alpha_n r(p,\gamma)^n =1, \qquad \sum_{n=1}^\infty n \alpha_n r(p,\gamma)^n = \mu(p,\gamma).
\end{equation*}
Thus we may consider $p_n: = \alpha_n r(p,\gamma)^n$ as a probability distribution on $\N$,  with finite expectation $\mu(p,\gamma)$. 
We have 
\begin{equation*}
	u_N:= C_N\, r(p,\gamma)^N = \sum_{n_1+ \cdots +n_D =N } p_{n_1} \times \cdots \times p_{n_D}.
\end{equation*}
This is the probability for a renewal process with waiting time distribution $(p_n)$ to have a renewal point at $N$, given that it had one at $0$. 
For our purpose it is more convenient to view $(p_n)$ as a measure on $p\N$. The renewal process is then called $p$-periodic.

The weights $p_N(\mathcal{X})$ from Eq.~\eqref{eq:weight} become
\begin{equation*} \label{eq:conditioned}
	p_N(\mathcal{X}) = p_{n_1} \cdots p_{n_D} / u_N
\end{equation*}
where $\mathcal{X}$ is a partition of $\{0,1,\ldots,pN-1\}$ into $D$ consecutive intervals of lengths $pn_1,\ldots,pn_D$. 
Thus $\P_N$ is a $p$-periodic 
renewal process conditioned on $0$ and $pN$ being renewal points. 

Now, given a distribution on $\N$ with finite expectation $\mu$, there is a standard way of defining a stationary renewal process, or, in our case, a 
$p$-periodic renewal process. To each partition of $\Z$ we associate the indicator function of the renewal points. In this way $\mathcal{P}_\infty$ 
inherits the product topology and Borel $\sigma$-algebra from $\{0,1\}^\Z$. There is a unique measure $\P$ on $\mathcal{P}_\infty$ 
such that for all integer $j$, 
\begin{equation*}
	\P( pj \ \text{is a renewal point}) = \frac{1}{\mu(p,\gamma)}, 
\end{equation*}
and for all $j<k$, 
\begin{equation*}
	\P\bigl( pk \ \text{is a renewal point} \mid pj \ \text{is a renewal point} \bigr)  = u_{k-j}.
\end{equation*}
 This measure $\P$ is invariant with respect to shifts by multiples of $p$. The following three lemmas will be  fundamental for the 
investigation of the thermodynamic limit.

 \begin{lemma}[$\P_N\to \P$] \label{lem:pntop}
Let 
$r,s \in \N$  have distance $\geq d$ from $0$ and $N$,  $d \leq r,s \leq N - d $. Consider the event $\mathcal{E}_{rs}$ that 
$pr$ and $ps$ are renewal points of $\mathcal{X}$ and moreover the restriction of $\mathcal{X}$ to 
$\{pr,\ldots,ps-1\}$ coincides with some given partition of this subset. Then
\begin{equation*}
	\bigr| \P_N(\mathcal{E}_{rs}) - \P(\mathcal{E}_{rs}) \bigl|  \leq \const
		 \P_N(\mathcal{E}_{rs}) \sup_{k \geq d} |u_k - \mu^{-1} |. 
\end{equation*}
\end{lemma}

\begin{proof}
	Suppose that we require that $\mathcal{X}$ coincides, in $\{pr,\ldots,ps-1\}$, with a partition into 
	successive intervals of lengths $pn_1,\ldots, p n_D$, $n_1+ \cdots +n_D= s-r$. Then 
	\begin{align*}
		\P(\mathcal{E}_{rs}) &= \mu^{-1} p_{n_1} \cdots p_{n_D}  \\
		\P_N(\mathcal{E}_{rs}) & = u_r p_{n_1} \cdots p_{n_D} u_{N-s} / u_N \\
		\frac{\P_N(\mathcal{E}_{rs})- \P(\mathcal{E}_{rs})}{\P_N(\mathcal{E}_{rs})} & = 
			1-  \frac{u_N } {\mu u_r  u_{N-s} }
	\end{align*}
	The claim then follows from the observation that $u_n \to \mu^{-1}$ as  $n\to \infty$. 
\end{proof}

\begin{lemma}[Long intervals are unlikely] \label{lem:long-intervals}
	Let $\alpha,\beta \in \Z$ with $\beta - \alpha \geq p d$ for some $d\in \N$ and $c(p,\gamma)>0$ as in Prop.~\ref{lemma:sub-mult}.  
	Then, for all $N\in \N$, 
	\begin{equation*}
		\P_N\bigl( \mathcal{X}\text{ has no renewal point in } \{\alpha,\ldots,\beta-1\}\bigr) \leq c(p,\gamma) \sum_{k\geq d} k p_k. 
	\end{equation*}
\end{lemma}

\begin{proof}
	The submultiplicativity Prop.~\ref{lemma:sub-mult} gives 
	\begin{equation*}
		\frac{u_j u_{N-j-n} }{u_N} \leq  \frac{1}{u_n} \leq c(p,\gamma).   
	\end{equation*}
	The probability that a partition $\mathcal{X}$ contains the interval $\{pj,\ldots,pj+pn-1\}$ 
	is therefore
	\begin{equation*}
		u_j p_n u_{N-j-n}/ u_N \leq c(p,\gamma) p_n.
	\end{equation*}
	We have to sum over pairs $(j,n)$ such that $p j \leq \alpha$ and $pj + pn \geq \beta$. Noting that 
	\begin{equation*}
		 \sum_{j \geq 0} \sum_{n \geq j+d} p_k =   \sum_{n \geq d} (n-d) p_n \leq  \sum_{n \geq d} n p_n,
	\end{equation*}	
	we obtain the desired inequality. 
\end{proof} 

\begin{lemma}[The renewal process is clustering] \label{lem:renewal-clust}
	Let $0\leq r\leq s\leq N$ have mutual distance $s-r \geq d$. Let $\mathcal{L}$ be the event that 
	$pr$ is a renewal point and  moreover the restriction of $\mathcal{X}$ 
	to $\{0,\ldots, pr-1\}$ coincides with some given partition of this  subset. Define $\mathcal{R}$ in a similar way 
	referring to the subset $\{ps,\ldots,pN-1\}$. Then 
	\begin{equation*}
	\bigr| \P_N(\mathcal{L} \cap \mathcal{R}) - \P_N(\mathcal{L}) \P_N(\mathcal{R})  \bigl|  \leq \const
		   \P_N(\mathcal{L}) \P_N(\mathcal{R}) \sup_{k \geq d} |u_k - \mu^{-1} |.
	\end{equation*}
\end{lemma}

\begin{proof}
	The proof is similar to the proof of Lemma~\ref{lem:pntop}, except that the term to be estimated, in the end, is 
	\begin{equation*}
		\frac{u_{s-r} / u_N}{ (u_{N-s} /u_N)\, (u_r /u_N)} - 1 = \frac{u_{s-r} u_N}{u_{N-s}u_r} - 1. 
	\end{equation*} 
	Again, this term is small for large $d$ because $u_n \to \mu^{-1}$ as $n\to \infty$. 
\end{proof}

\section{Correlation functions} \label{sec:correlations}

In order to prove our main results, all there is left to do now is to go from the convergence of measures $\P_N \to \P$ to convergence of states $\la a \ra_N \to \la a \ra$. This is essentially an interchange in the order of summation and limits, and, therefore, involves some technical estimates; we hope, however, that we have conveyed the simplicity of the underlying idea. 

The proofs follow ideas from~\cite{jls}, even though the presentation is slightly different as no use was made, in that work, of the framework of quasi-state decompositions. Some estimates become more involved because we extend the results to bosons, so that creation and annihilation operators are unbounded. Another novelty is the proof of the insensitivity towards the precise choice of the domain of integration in Sect.~\ref{sec:domchoice}.

\subsection{A variant of the Cauchy-Schwarz inequality}

Recall Eq.~\eqref{eq:quasi-decomp}
\begin{equation} \label{eq:quasi-decomp-bis}
    \la a \ra_N = \frac{1}{C_N} \la \Psi_N, a \Psi_N \ra = \sum_{\mathcal{X} \in \mathcal{P}_N} p_N(\mathcal{X}) \omega_\mathcal{X}(a).
\end{equation}
 We would like to deduce from Lemma~\ref{lem:pntop} that  
\begin{equation} \label{eq:candidate-limit}
  \la a \ra_N \to  \int \dd \P(\mathcal{X}) \omega_\mathcal{X}(a) =: \la a \ra. 
\end{equation} 
To this aim we will need to control the contribution of unfavorable 
partitions $\mathcal{X}$. When $a$ is a diagonal, bounded operator, and  $\mathcal{F}$ is some collection of partitions, we can write
\begin{equation} \label{eq:estimate-1}
	 \sum_{\mathcal{X} \in \mathcal{F}} p_N(\mathcal{X})\, |\omega_\mathcal{X}(a)| 
      \leq ||a||\, \P_N( \mathcal{X} \in \mathcal{F}).
\end{equation}
For observables that are not  diagonal or 
bounded, we use a variant of the Cauchy-Schwarz inequality $\Tr (\rho A^*B) \leq (\Tr \rho A^*A)^{1/2} (\Tr \rho B^*B)^{1/2}$.  
This will help us, in some situations, replace the estimate~\eqref{eq:estimate-1} by an estimate of the type 
\begin{equation*}
	\sum_{\mathcal{X} \in \mathcal{F}} p_N(\mathcal{X})\, |\omega_\mathcal{X}(a)| 
      \leq d(a) \la a ^*a \ra_N^{1/2} \bigl( \P_N( \mathcal{X} \in \mathcal{F}) \bigr)^{1/2}
\end{equation*}	
for some $a$-dependent constant $d(a)$. 

\begin{lemma} \label{lem:cauchy-schwarz}
  Let $\rho$ be a trace class, positive operator in $\ell^2(\N)$ with matrix $(\rho_{ij})_{i,j\in \N}$. 
  Let $A= (A_{ij})_{i,j\in \N}$ be an operator in $\ell^2(\N)$. 
    Suppose that $\Tr \rho A^* A <\infty$, with $A$ possibly unbounded, and that $A$ 
	has at most $d(A)$  non-zero matrix elements $A_{ij}$ per row. Then for all $S \subset \N$, 
	   \begin{equation*}
		\sum_{i \in \N} \sum_{j \in S}   |\rho_{ij} A_{ji}|  \leq d(A) \Bigl( \Tr \rho A^*  A\Bigr)^{1/2}  
		    \Bigl( \Tr \rho \mathbf{1}_S \Bigr)^{1/2}.
	   \end{equation*} 
\end{lemma}

\begin{proof}  It is enough to treat the case when $A$ has at most one non-zero matrix element per row. 
  Write $\rho_{ij} = \sum_{n \in \N} \lambda_n \overline{x_i^n} x_j^n$ with 
  $\lambda_1 \geq \lambda_2 \geq \cdots \geq 0$ the eigenvalues of $\rho$ 
  and $x^n \in \ell^2(\N)$ the normalized eigenvectors.  
 There is a  map $\phi: \N\to \N$ assigning to each row $i$ the column $j= \phi(i)$ with the non-zero entry. If all entries in the row vanish, 
$\phi(i)$ is arbitrary. 
  Then
  \begin{align*}
    \sum_{i \in \N} \sum_{j \in S} |\rho_{ij} A_{ji}| 
      & \leq \sum_n  \sum_{j \in S} \lambda_n |x_j^n A_{j,\phi(j)} x_{\phi(j)}^n| \\
      & \leq \sum_n \lambda_n 	\Bigl( \sum_{i \in S}  |x_i^n|^2 \Bigr)^{1/2} \Bigl( \sum_{i \in S} |A_{i,\phi(i)}x_{\phi(i)}^n|^2 \Bigr)^{1/2} 
	   \\      
      & \leq \left( \sum_n \lambda_n \sum_{i \in S} |x_i^n|^2 \right)^{1/2}
	\left( \sum_n \lambda_n \sum_{i \in S}  |A_{i,\phi(i)} x_{\phi(i)}^n |^2\right)^{1/2}
	 \\
      & \leq  \Bigl( \sum_{i \in S} \rho_{ii}\Bigr)^{1/2}\,\Bigl( \Tr \rho A^* A \Bigr)^{1/2}. \qedhere
  \end{align*}
\end{proof}

As a first application of the previous lemma we show: 
\begin{prop} \label{prop62}
  Let $a = c_{k_1}^*\cdots c_{k_n}^*c_{q_t} \cdots c_{q_1}$  be a  Wick ordered monomial 
  of creation and annihilation operators. Then, for all $N \in \N$, 
  \begin{equation*}
    	  \sum_{\mathcal{X} \in \mathcal{P}_N} p_N(\mathcal{X}) |\omega_\mathcal{X}(a)| \leq \la a^* a \ra_N^{1/2}.
  \end{equation*}
\end{prop}

\begin{proof}
  For every $\vect{m} = (m_1,\ldots,m_N)$ with $0 \leq m_1\leq \cdots \leq m_N \leq pN-p$, 
  there is at most one $\vect{m'}$ such that $\la \psi_{\vect{m'}},a \psi_{\vect{m}} \ra \neq 0$. 
   Thus we are in the setting of case 2. of the previous lemma, with $d(a) =1$, and we deduce
  \begin{equation*}
      \sum_{\mathcal{X} \in \mathcal{P}_N} p_N(\mathcal{X}) |\omega_\mathcal{X}(a)| 
       \leq \frac{1}{C_N} \sum_{\vect{m},\vect{m'}} \Bigl| a_N(\vect{m}) a_N(\vect{m'}) \la \psi_{\vect{m'}},a \psi_{\vect{m}} \ra \Bigr| 
       \leq \la a^* a \ra_N^{1/2}. \qedhere
  \end{equation*}
\end{proof}
As a consequence, the candidate limit $\la a \ra$ in Eq.~\eqref{eq:candidate-limit} is well-defined:
\begin{cor} \label{cor:well-defined}
	Let $a$ be a linear combination of products of creation and annihilation operators $c_k^*$, $c_m$. 
	The integral from Eq.~\eqref{eq:candidate-limit} is absolutely convergent,
	$$\int_{\mathcal{P}_\infty} \dd \P(\mathcal{X})\, | \omega_\mathcal{X}(a)| < \infty.$$ 
\end{cor}
This is shown by passing to the limit $N\to \infty$ in the previous proposition, and using Lemma~\ref{lem:pntop} and 
the uniform moment bounds from Prop.~\ref{prop:corr-bounds}.

\begin{remark}
	Cor.~\ref{cor:well-defined} has an analogue when $a$ is a local, bounded operator, for example $a = \exp( i c_1^*c_1)$. This can be shown 
	with the help of yet another variant of the Cauchy-Schwarz inequality. 
\end{remark}
	
Note that 
despite the notation $\la \cdot \ra$, we do not know yet whether the linear map $a \mapsto \la a \ra$ is positive or defines a proper 
state.

\subsection{Thermodynamic limit}

In this section we prove Theorem~\ref{thm:thermolim} when the infinite cylinder $\Lambda = \mathcal{Z}$ is chosen as the domain of integration. We first estimate the difference between left- and right-hand side in Eq.~\eqref{eq:candidate-limit}. 

\begin{lemma} \label{prop:absolute1}
    Let  $a$ be a  Wick ordered monomial supported in  $ \{\alpha,\ldots, \beta-1\} \subset \Z$. 
  Let $d\in \N$ and   $\mathcal{E}_d$  be the set of partitions with  renewal points in both  
	$[\alpha-pd,\alpha]$ and $[\beta,\beta+pd]$. Then 
	\begin{equation*}
	\sum_{\mathcal{X} \notin \mathcal{E}_d} p_N(\mathcal{X}) |\omega_\mathcal{X}(a)| 
	 \leq \la a^* a\ra_N ^{1/2}\,   
    		\Bigl( 2 c(p,\gamma) \sum_{k\geq d} k p_k \Bigr)^{1/2}.
    	\end{equation*} 	
\end{lemma}

\begin{proof} 
	We abbreviate $\vect{m} = (m_1,\ldots,m_N)$ and 
	\begin{equation*}
		a_N(\vect{m})= a_N(m_1,\ldots,m_N), \quad \psi(\vect{m}) \widehat{=} \psi_{m_1}(z_1) \cdots \psi_{m_N}(z_N).
	\end{equation*}
  Let $\rho(\vect{m},\vect{m'}) : = a_N(\vect{m}) a_N(\vect{m'}) / C_N$. 
  For $d\in \N$, let $S$ be the set of $\vect{m}$'s without renewal points in 
  $[\alpha-pd,\alpha]$ or without renewal point in $[\beta,\beta+pd]$. 
For every  
  $\vect{m}$, $ \la \psi(\vect{m'}),a \psi(\vect{m})\ra$ is non-vanishing for at most one 
  vector $\vect{m'}$. Thus by Lemma~\ref{lem:cauchy-schwarz},
  \begin{align*}
	\sum_{\mathcal{X} \notin \mathcal{E}_d} p_N(\mathcal{X}) |\omega_\mathcal{X}(a)| 
    & \leq \frac{1}{N! C_N} \sum_{\vect{m},\vect{m'} \in S} 
	\Bigl| a_N(\vect{m}) a_N(\vect{m'}) \la \psi(\vect{m}),a \psi(\vect{m'})  
	  \ra \Bigr| \\	
        &  \leq \la a^* a\ra_N ^{1/2}\,   
    \Bigl[ \sum_{\vect{m} \in S} |a_N(\vect{m})|^2 /(N!C_N )\Bigr]^{1/2} \\
	& = \la a^* a\ra_N ^{1/2}\, \bigl(\P_N(\mathcal{E}_d^\mathrm{c}) \bigr)^{1/2}
\end{align*}  
and we conclude with Lemma~\ref{lem:long-intervals}. 
\end{proof}

\begin{lemma} \label{lem:aux-thermo}
	Let  $a$ be a  Wick ordered monomial supported in  $ \{\alpha,\ldots, \beta-1\} \subset \Z$, $d \in \N$, 
	and $\mathcal{E}_d$ as in the previous lemma. We may consider  $\mathcal{E}_d$ as a subset of $\mathcal{P}_N$ or of 
	$\mathcal{P}_\infty$. Then  
    	\begin{multline*}
		\Bigl| \sum_{\mathcal{X} \in \mathcal{E}_d} p_N(\mathcal{X})\omega_\mathcal{X}(a) 
			- \int_{\mathcal{E}_d} \dd \P(\mathcal{X}) \omega_\mathcal{X}(a) \Bigr| \\
			\leq  \const \Bigl( \sup_{k\geq d}|u_k - \mu^{-1}| \Bigr) \sum_{\mathcal{X}\in \mathcal{P}_N} p_N(\mathcal{X}) |\omega_\mathcal{X}(a)|.
	\end{multline*}
\end{lemma}
This is a consequence of Lemma~\ref{lem:pntop} and the locality of the $\omega_\mathcal{X}$.
\begin{proof}
	A partition $\mathcal{X} \in \mathcal{E}_d$ has a renewal point $r \in [\alpha - pd, \alpha]$ and a renewal 
	point $s \in [\beta,\beta+pd]$. Choose $r$ the largest possible and $s$ the smallest possible. The value of 
	$\omega_\mathcal{X}(a)$ depends only on $\mathcal{X}_{r,s}$, the restriction of the partition to the volume enclosed 
	by the renewal points. The contribution of partitions with the same $r$, $s$ and same restriction $\mathcal{X}_{r,s}$ 
	is of the form 
	$\omega_{\mathcal{X}_{r,s}} (a) \P_N( \mathcal{E}_{rs})$. 
	A similar form can be derived for partitions of the infinite lattice $\Z$. 
	By Lemma~\ref{lem:pntop}, 
	\begin{multline*}
		\Bigl| \omega_{\mathcal{X}_{r,s}} (a) \P_N( \mathcal{E}_{rs}) - \omega_{\mathcal{X}_{r,s}} (a) \P( \mathcal{E}_{rs}) \Bigr| \\
			\leq  \const \Bigl( \sup_{k\geq d}|u_k - \mu^{-1}| \Bigr)\, \bigl|\omega_{\mathcal{X}_{r,s}} (a)\bigr| \, \P_N(\mathcal{E}_{rs}),
	\end{multline*} 
	and the proof is concluded by summing over $r$, $s$ and $\mathcal{X}_{rs}$. 
\end{proof}

The previous two lemmas, together with Prop.~\ref{prop62} and Lemma~\ref{lem:pntop} yield a lattice version of Theorem~\ref{thm:thermolim}.
\begin{cor} \label{cor:a-aN}
Let $a$ be a Wick-ordered monomial whose support has distance $\geq pd$ to the lattice boundaries, 
$\supp a \subset \{pd,\ldots, pN - pd-1\}$. Then 
\begin{equation*}
	\Bigl|\la a \ra_N - \la a \ra \Bigr| \leq \const \Bigl(  \sup_M \la a^*a\ra_M \Bigr)^{1/2}  
			\Bigl( 4 c(p,\gamma) \bigl( \sum_{k\geq d} k p_k\bigr)^{1/2} + \sup_{k\geq d} |u_k - \mu^{-1}| \Bigr).
\end{equation*}   
\end{cor}
Note that the upper bound goes to $0$ as $d\to \infty$ because $\sum k p_k <\infty$, $u_k\to \mu^{-1}$, and because of
the uniform moment bounds from Prop.~\ref{prop:corr-bounds}. 

\begin{proof}[Proof of Theorem~\ref{thm:thermolim} when $\Lambda=\mathcal{Z}$]  \label{proof:thermolim}
In view of Cor.~\ref{cor:a-aN}, all there is left to do is to go from  lattice correlations to continuum correlations. This is easily achieved, thanks 
to the explicit relation between lattice and continuum and the Gaussian decay of the one-particle functions $\psi_k(z)$. 
For the sake of clarity we write down the proof for the two-point functions only; the other correlation functions can be treated in a similar way. Recall 
\begin{equation} \label{eq:two-point}
	\rho_2^{N,\mathcal{Z}}(z_1,z_2;z'_1,z'_2) 
		=  \sum_{ 0\leq k,\ell,m,n \leq p N - p} \la c_m^* c_n^* c_\ell c_k \ra_N 
				\, \psi_k(z_1)\, \psi_\ell(z_2)
			\, \overline{\psi_m(z'_1)}\, \overline{\psi_n(z'_2)}.
\end{equation}
Therefore we define 
\begin{equation*}
	\rho_2(z_1,z_2;z'_1,z'_2) 
		:=\sum_{  k,\ell,m,n \in \Z} 
			\la c_m^* c_n^* c_\ell c_k \ra
				\, \psi_k(z_1)\, \psi_\ell(z_2)
			\, \overline{\psi_m(z'_1)}\, \overline{\psi_n(z'_2)}
\end{equation*}
with the bulk expectation $\la \cdot \ra$ as in Eq.~\eqref{eq:candidate-limit}. The infinite sum converges because 
the expectations $\la \cdots \ra$ 
appearing in the sum are bounded  as in  Prop.~\ref{prop:corr-bounds}, and 
\begin{equation*}
	 \sup_z \sum_{m\in \Z} |\psi_m(z)| \propto \sup_x \sum_{m\in\Z} 
			\exp[- (x-m\gamma)^2/2] <\infty.
\end{equation*}
The translational covariance of $\omega_\mathcal{X}(a)$ and  the stationarity of the renewal process $\P$ make the two-point correlation function $p \ell^2/R$-periodic in the sense of 
Eq.~\eqref{eq:periodic}. 

The difference between the finite volume and infinite volume two-point function is a sum over integers $k,l,m,n$. 
Suppose that $z_1,z'_1, z_2, z'_2$ are at distance  $\geq D \geq 2 p d \gamma$ from 
the boundaries of the cylinder $x=0$ and $x= pN\gamma$.
 The contribution to the two-point function 
from summands with all of the four indices between $d$ and $pN - pd$ can be bounded with the help of Cor.~\ref{cor:a-aN}. 
The contribution from quadruplets $(k,l,m,n)$ with  $k \leq pd$, is bounded by a constant times 
\begin{equation} \label{eq:gauss-tail}
	 \sum_{k \leq d} |\psi_k(z'_1)| \propto \sum_{k \leq pd} \exp \bigl( - (x'_1- k \gamma)^2 /2 \bigr) \leq 
			\sum_{\kappa \geq 0} \exp\bigl( - p^2 ( d + \kappa )^2 \gamma^2/2\bigr). 
\end{equation}
which is small for large $d$. Contributions where another index is smaller than $pd$ or larger than $pN-pd$ can 
be bounded in a similar way. 
\end{proof}

	
\subsection{Clustering}

The clustering for $\Lambda = \mathcal{Z}$ is deduced from the renewal process Lemma~\ref{lem:renewal-clust},  in the 
same way as the thermodynamic limit was deduced, in the previous section, from Lemma~\ref{lem:pntop}. Again, we start with lattice correlations.

Consider $a$ and $b$  two Wick-ordered monomials of creation and annihilation operators 
with supports at mutual distance $\geq 3pd$ for some $d \in \N$. Thus let $\alpha, \beta \in \Z$ 
such that $\beta- \alpha \geq 3 d $ and   
$\supp a \subset \{\ldots, p\alpha - 1\}$ and $\supp b \subset \{p\beta, o\beta+1,\ldots\}$.  Let $\mathcal{F}$ be the set of partitions $\mathcal{X}$ 
with a renewal point in $[p\alpha,p\alpha +pd]$ and $\mathcal{G}$ 
the set of partitions with a renewal point in $[p\beta- pd,p\beta]$.  

We write the expectation of $ab$ as 
a sum over partitions in $\mathcal{F}\cap \mathcal{G}$ plus a remainder; similarly for $a$ and $b$. In order to estimate $\la a b\ra_N - \la a \ra_N \la b\ra_N$, we need to estimate four terms,  the three remainders, 
and the difference 
\begin{equation} \label{eq:clust-main-diff}
	\sum_{\mathcal{X} \in \mathcal{F}\cap \mathcal{G}} p_N(\mathcal{X}) \omega_\mathcal{X}(a) \omega_ \mathcal{X}(b) 
		- \Bigl(  \sum_{\mathcal{X} \in \mathcal{F}} p_N(\mathcal{X}) \omega_\mathcal{X}(a) \Bigr) 
			 \Bigl(  \sum_{\mathcal{X} \in \mathcal{G}} p_N(\mathcal{X}) \omega_\mathcal{X}(b) \Bigr). 
\end{equation}		
Recall that if 
$\mathcal{X} \in \mathcal{F}\cap \mathcal{G}$, then $\omega_{\mathcal{X}}(ab) 
	= \omega_{\mathcal{X}}(a)\omega_{\mathcal{X}}(b)$.
The remainders for $a$ and $b$ are simplest to estimate,  
\begin{equation*}
	\sum_{\mathcal{X}\notin \mathcal{F}} p_N(\mathcal{X}) \bigl| \omega_\mathcal{X}(a) \bigr| \leq 
		\Bigl( \la a ^*a \ra_N  c(p,\gamma) \sum_{k \geq d} k p_k \Bigr)^{1/2} ,
\end{equation*}
and an analogous inequality holds for $b$ and $\mathcal{G}$. 
The proof is similar to the proof of  Lemma~\ref{prop:absolute1}. 
The difference~\eqref{eq:clust-main-diff} is bounded by some constant times 
\begin{equation*}
	 \Bigl( \sup_{k \geq d} | \mu^{-1} u_k - 1| \Bigr) 
			\la a ^*a \ra_N^{1/2}\, \la b^*b \ra_N^{1/2}.
\end{equation*}
This is shown with the help of Lemma~\ref{lem:renewal-clust}, and proceeding in a way similar to the proof of Lemma~\ref{lem:aux-thermo}.

The remaining estimate of the contribution to $\la a b \ra_N$ from partitions not in $\mathcal{F}\cap \mathcal{G}$ is
slightly more involved. 
First we switch to an occupation number picture. With $\vect{m}=(m_1,\ldots,m_N)$, we associate a sequence of 
occupation numbers $\vect{n} = (n_0,n_1,\ldots,n_{pN-p})$ in the obvious way; 
for example, when $N=2$, $p=2$, and  $(m_1,m_2) = (1,1)$, we have $(n_0,n_1,n_2,n_3) =(0,2,0,0)$. 
We let $|\vect{n}\ra$ be the normalized wave function proportional to the (anti-)symmetrized  product of  
$\psi_{m_1}$,..., $\psi_{m_N}$. The $| \vect{n}\ra$'s form an orthonormal system.  
The many-particle wave function becomes 
\begin{equation*}
	\Psi_N = \sum_{\vect{n}} A_N(n_0,n_1,\ldots,n_{pN-p}) | n_0 n_1\cdots n_{pN-p} \ra.
\end{equation*}
When $p$ is odd, the coefficients $A_N(\vect{n})$ are in one-to-one correspondence with the coefficients 
$a_N(\vect{m})$. When $p$ is even, the correspondence is up to factors $\sqrt{n_i!}$.

Next, we observe that the notion of renewal point, since it does not depend on the order of the $m_j$'s, 
can be transferred to occupation numbers.
Hence $pk$ is a renewal point of $\vect{n}$ if and only if 
\begin{equation} \label{eq:renewal-nj}
	\sum_{j=0}^{pk-1} n_j = k\quad \text{and}\quad \sum_{j=1}^{pk-1} j n_j = pk(k-1)/2.
\end{equation}
The admissibility condition for $\vect{m}$ leads to the following property, valid whenever $A_N (\vect{n} ) \neq 0$: 
for all $k =1,\cdots, N-1$,  
\begin{equation} \label{eq:admissible-nj}
	\sum_{j=0}^{pk-1} n_j = k\ \Rightarrow \sum_{j=1}^{pk-1} j n_j \geq pk(k-1)/2.
\end{equation}
Now let $S$ be the set of $\vect{n}$'s with no renewal 
point in $[\alpha,\alpha+pd]$ or no  renewal point in $[\beta-pd,\beta]$, and 
 $E$ the set of \emph{pairs} $(\vect{n}, \vect{n'})$ with no \emph{common} renewal 
point in $[\alpha,\alpha+pd]$ or no common renewal point in $[\beta-pd,\beta]$. 
Thus $\vect{n} \in S$ is a sufficient, but not necessary, condition for $(\vect{n},\vect{n'})$ to be in $E$. 
Write 
\begin{equation} \label{eq:Esum}
	 \sum_{\mathcal{X} \notin \mathcal{F}\cap \mathcal{G}} p_N( \mathcal{X}) \omega_\mathcal{X}(ab)  
		= \frac{1}{C_N}  
		\sum_{ (\vect{n}, \vect{n'}) \in E}  A_N(\vect{n})A_N(\vect{n'}) 
			\bigl \la \vect{n}  |\, a b\, | \vect{n'}\bigr \ra.
\end{equation}
The next lemma paves the way for an application of a Cauchy-Schwarz-inequality to Eq.~\eqref{eq:Esum}.
\begin{lemma}
	Let $a$, $b$, $E$ as described above, and suppose in addition that $a$ preserves the total particle number, i.e., it is the product 
	of $n$ creation operators and the same number $n$ of annihilation operators. 
	If $(\vect{n},\vect{n'}) \in E$ and  
	\begin{equation} \label{eq:non-vanish}
	A_N(\vect{n})A_N(\vect{n'}) 
			\bigl \la \vect{n} | a b| \vect{n'} \bigr \ra \neq 0, 
	\end{equation}
	then $\vect{n} \in S$ or $\vect{n'} \in S$.
\end{lemma}

\begin{proof} 
Recall the notion of support of an observable: when $a = c_1^* c^*_3$, $\supp a = \{1,3\}$. If  Eq.~\eqref{eq:non-vanish} holds, 
 then $\vect{n}$ and $\vect{n'}$ must coincide outside $\supp(a) \cup \supp(b)$. 
We may assume, without loss of generality, that 
\begin{equation} \label{eq:order-nj}
	 \sum_{j=1}^{p \alpha - 1} j  n_j \geq \sum_{j=1}^{p \alpha - 1} j  n'_j.
\end{equation}
(Otherwise swap $\vect{n}$ and $\vect{n'}$.)
If $a$ preserves the total particle number,  e.g., $a = c_1^*c_2$, then $\vect{n}$ and $\vect{n'}$ must have the same 
number of particles in $\supp a$. It follows that for every site $v$ between $\supp a$ and $\supp b$, they have the same number of 
particles to the left of $v$, 
\begin{equation} \label{eq:left-nj}
	n_0 + n_1+ \cdots + n_{v-1} = n'_0 + \cdots + n'_{v-1}, \qquad  p \alpha \leq v \leq p\beta - 1.
\end{equation}
Suppose that $\vect{n}$ has 
a renewal point $pk$ between $\supp a$ and $\supp b$, $\alpha \leq k \leq \beta $. 
Then $\vect{n}$ has $k$ particles to the left of $v = pk$, and by Eq.~\eqref{eq:left-nj}, so has $n'_j$. 
Moreover, since $\vect{n}$ and $\vect{n'}$ coincide outside $\supp(ab)$, Eq.~\eqref{eq:order-nj} implies that 
\begin{equation*}
	\sum_{j=1}^{pk-1} j n'_j \leq \sum_{j=1}^{pk-1} j n_j = pk (k-1) /2.
\end{equation*}
Eqs.~\eqref{eq:admissible-nj} and~\eqref{eq:renewal-nj} then imply that $pk$ is a renewal point of $\vect{n'}$ too. 
Thus every renewal point of $\vect{n}$ between $\supp a$ and $\supp b$ is in fact a \emph{common} renewal point 
of $\vect{n}$ and $\vect{n'}$. Therefore, if $(\vect{n},\vect{n'}) \in E$, necessarily $\vect{n} \in S$. 
\end{proof}

Now we can apply  Lemma~\ref{lem:cauchy-schwarz} which together with Lemma~\ref{lem:long-intervals} 
yields 
\begin{equation*}
	\Bigl| \sum_{\mathcal{X} \notin \mathcal{F}\cap \mathcal{G}} p_N(\mathcal{X}) \omega_\mathcal{X}(ab) \Bigr| 
		\leq 2 \la (ab)^*ab\ra_N ^{1/2}  \Bigl( 2 c(p,\gamma) \sum_{k \geq d} k p_k \Bigr)^{1/2}.
\end{equation*} 
Using our uniform moment bounds, we obtain: 
\begin{prop} \label{prop:lattice-clust}
	Let $a$ be a product of $n$ creation operators and $n$ annihilation operators, and $b$ a product of 
	$m$ creation and $m$ annihilation operators. 
	Suppose that $\supp a \subset \{\ldots, p\alpha -1\}$, 
	$\supp b \subset \{p \beta,\ldots \}$ and $\beta - \alpha \geq 3d$.  Then, for  some suitable constant $K_{m+n}$ and for all $N$,  
	\begin{equation*}
		\Bigl| \la a b \ra_N - \la a \ra_N \la b \ra_N \Bigr| \leq   K_{m+n} \Bigl( (\sum_{k \geq d} k p_k)^{1/2} + \sup_{k \geq d} | u_k - \mu^{-1}| \Bigr).
	\end{equation*}
\end{prop}

\begin{remark}[Off-diagonal decay]
		Suppose, for example, $a = c_k^*$ and $b= c_l^* c_m c_n$ for some $k,l,m,n \in \N$. 
		Then $\la a \ra_N = 0 = \la b \ra_N =0$ and, because of $y$-momentum conservation,  
		$ \la a b \ra_N  =0$ unless $k+l = m+n$. Imagine shifting $b$ along the $x$ axis, 
		\begin{equation*}
		 	b \to c_{l+d}^* c_{m+d} c_{n+d} =\tau_x^d(b).
		\end{equation*}
		Then $\la a \tau_x^d(b) \ra_N \neq 0$ unless $k+( l+d) = m+n +2d$, so we find 
		\begin{equation*}
			d > k + l - m - n \ \Rightarrow \la a \tau_x^d(b)\ra_N = 0 = \la a \ra_N \la \tau_x^d(b) \ra_N.
		\end{equation*} 
		The same reasoning applies to higher correlations where $a$, or $b$, does not conserve particle number.
		Hence there is clustering for off-diagonal correlations too. 
\end{remark}

\begin{proof}[Proof of Theorem~\ref{thm:clustering}] 
	Just as in the proof of Theorem~\ref{thm:thermolim} on p.~\pageref{proof:thermolim}, all we need to do is pass from lattice 
	(Prop.~\ref{prop:lattice-clust})
	to continuum. Again, for simplicity we write down the proof only for the two-point correlation. Let $z_1,z'_1$ have large distance 
	 from $z_2,z'_2$ along the cylinder axis, i.e.,  $ x_1, x'_1 \leq w \gamma$ and $x_2,x'_2 \geq (w +3 D)\gamma$ for some integer $w$ and large $D$.   
	We write the sum in Eq.~\eqref{eq:two-point} as a main contribution $M$, plus a remainder. The main contribution consists of those summands where 
	$k,m \leq w +D$ and $l,n \geq w+2D$. Because of Prop.~\ref{prop:lattice-clust}, we have a bound 
	\begin{equation*}
		\bigl| \la c_m^* c_n^* c_l c_k \ra_N  - \la c_m^* c_k\ra_N \la c_n^*c_l\ra_N \bigr| \leq \const g(D)
	\end{equation*}
	for some function $g(D) \to 0$ as $D\to \infty$.  It follows that the difference between the main contribution $M$ and 
	\begin{equation*}
		\Bigl(\sum_{k,m \leq w+D} \la c_m^* c_k \ra_N \psi_k(z_1)\, \overline{\psi_m(z'_1)} \Bigr)
		 \, \Bigl(\sum_{l,n \geq w+2D} \la c_n^* c_l \ra_N \psi_l(z_2)\, \overline{\psi_n(z'_2)} \Bigr)
	\end{equation*}
	can be bounded by a constant times $g(D)$ too. But in this last term we recognize the main contribution to $\rho_1^{N}( z_1;z'_1) \rho_1^N(z_2;z'_2)$; 
	so we are left with the remainders to estimate. This is easily achieved: 
	the remainders, for the one-particle matrix as well as for the two-point function, can be estimated by terms of the 
	type~\eqref{eq:gauss-tail}. 
\end{proof}

\subsection{Non-influence of the domain of integration} \label{sec:domchoice}

In this section we show that the precise choice of domain of integration does not affect bulk correlations.  
The idea  is to rewrite integrals over $\Lambda$ as integrals over the infinite cylinder $\mathcal{Z}$ with the indicator function of 
$\Lambda$ in the integrand, and then translate the indicator into a lattice operator. This will allow us to view the indicator 
function as a quantity that lives at the cylinder's boundaries, and to decouple this boundary perturbation from bulk correlations with the help of 
the state's clustering. We start with a simple computation. 
 Let $m_1,\ldots,m_N \in \Z$ (not necessarily ordered or distinct). Then 
\begin{equation*}
	\int_{\Lambda^N} |\psi_{m_1}(z_1)|^2 \cdots |\psi_{m_N}(z_N)|^2 \dd z_1 \cdots \dd z_N 
			 = ||\psi_{m_1} ||_\Lambda ^2 \cdots  ||\psi_{m_N} ||_\Lambda ^2 .
\end{equation*}
We can rewrite this as 
\begin{equation*}
	||\psi_{m_1} \otimes \cdots \otimes \psi_{m_N}||_\Lambda^2  = 
			\bigl \la \psi_{m_1} \otimes \cdots \otimes \psi_{m_N}, J_{N,\Lambda} \psi_{m_1} \otimes \cdots \otimes \psi_{m_N} \bigr \ra_\mathcal{Z},
\end{equation*}
using the diagonal operator in $L^2(\mathcal{Z}^N)$
\begin{equation*}
	J_{N,\Lambda}:\ \psi_{m_1} \otimes \cdots \otimes  \psi_{m_N} \mapsto \bigl( ||\psi_{m_1} ||_\Lambda ^2 \cdots  ||\psi_{m_N} ||_\Lambda ^2 \bigr)\, 
	\psi_{m_1} \otimes \cdots \otimes  \psi_{m_n}
\end{equation*}
(we set $J_{N,\Lambda}$ equal to $0$ in the orthogonal complement of the lowest Landau level, i.e., the space spanned by the $\psi_k(z)$). 
 The Fock space version of this operator, again denoted $J_{N,\Lambda}$,  is 
\begin{equation*}
	J_{N,\Lambda} =  \prod_{k=0}^{pN-p} \bigl( ||\psi_k||_\Lambda^2 \bigr)^{\hat n_k}.
\end{equation*}
In this product, only boundary terms, $k$ small or close to $pN-p$, contribute. Indeed, for $0 \leq k \leq p N-p$, 
and $\Lambda = [a,b] \times [0,2\pi R]$, 
\begin{equation} \label{eq:erfc}
	||\psi_k||^2_\Lambda = \frac{1}{\sqrt{\pi}} \int_a^{b} 
			e^{ - (x-k\gamma)^2 } \dd x = 
	1 - \eps_k - \delta_{pN-p-k}.
\end{equation}
The error terms $\eps_k$ and $\delta_{j}$ depend on the precise choice of the domain of integration. They are small when $k \to \infty$, resp. $j\to \infty$. 
For example, when $a = 0$, $b = (pN-p)\gamma$, 
\begin{equation*}
	\delta_k = \eps_k  = \frac{1}{\sqrt{\pi}} \int_{k\gamma}^\infty e^{-s^2} \dd s \to 0 \quad (k\to \infty). 
\end{equation*}

\begin{lemma}
	 The normalization and the one-particle 
	density for $\Psi_N$ with domain of integration $\Lambda$ are given by
	\begin{align}
	\notag	||\Psi_N||_\Lambda^2 & = \la \Psi_N, J_{N,\Lambda} \Psi_N \ra_\mathcal{Z} \\ 
		\rho_1^{N,\Lambda}(z;z') & = 
				\sum_{k=0}^{pN-p} \frac{\la \Psi_N, c_k^*\, J_{N,\Lambda} c_k 	\Psi_N \ra_\mathcal{Z} }
					{\la \Psi_N, J_{N,\Lambda}\, \Psi_N \ra_\mathcal{Z}}\, \psi_k(z)\, \overline{\psi_k(z')}, \qquad (z,z' \in \Lambda). 
			\label{eq:one-part-finite}
	\end{align}
\end{lemma}
Similar formulas hold for $n$-point correlations. 
\begin{proof}
	The formula for the  normalization is a consequence of the computations used to define $J_{N,\Lambda}$.
	For the one-particle matrix, we note 
	\begin{align*}
		 &(c_k \Psi_N)(z_2,\ldots,z_N)  = \sqrt{N} \int_\mathcal{Z}  \overline{\psi_k(z_1)}\, \Psi_N(z_1,z_2,\ldots,z_N) \dd z_1\\
					& \qquad =  (N-1)!^{-1/2} \sum_{m_2,\ldots,m_N} a_N(k,m_2,\ldots,m_N) \psi_{m_2}(z_2) \cdots \psi_{m_N}(z_N).
	\end{align*}
	Let $d(m_2,\ldots,m_N):=||\psi_{m_2}||_\Lambda^2 \cdots ||\psi_{m_N}||_\Lambda^2$. We obtain 
	\begin{equation*}
		\bigl\la \Psi_N, c_k^*\, J_{N,\Lambda} c_k \Psi_N  \bigr \ra_\mathcal{Z} 
			= \frac{1}{(N-1)!} \sum_{m_2,\ldots,m_N} |a_N(k,m_2,\ldots,m_N)|^2\,  d(m_2,\ldots,m_N). 
	\end{equation*}
	On the other hand, $N$ times the integral of ${\Psi_N(z,z_2,\ldots,z_N)} \overline{\Psi_N(z',z_2,\ldots,z_N)}$ with the $z_j$ 
	integrated over $\Lambda$, equals 
	\begin{equation*}
		N \frac{1}{N!} \sum_{m_1} |a_N(m_1,\ldots,m_N)|^2 d(m_2,\ldots,m_N) 				\psi_{m_1}(z) \overline{\psi_{m_1}(z')}, 
	\end{equation*}
	and the proof is easily concluded. 
\end{proof}

Because of Eq.~\eqref{eq:erfc}, it is natural to think of our lattice indicator $J_{N,\Lambda}$ as a product of a left, bulk, and right term. 
We write
$	J_{N,\Lambda}:= \mathcal{L} \mathcal{B} \mathcal{R}$
with 
\begin{equation*}
	\mathcal {L} = \prod_{j=0}^{d-1}  \bigl(1 - \eps_j -\delta_{pN- p -j}\bigr)^{\hat n_j},
\end{equation*}
$\mathcal{B}$ a similar product for $j$ from $d$ to $ pN-p-d$, and $\mathcal{R}$ the product from $pN-p-d $ to $ pN -p$. 
\begin{lemma} \label{lem:61}
	For all $N$ and $d$ with $pN - p \geq 3d$ and a suitable, $N$-independent function $f(d)$ with $f(d) \to 0$ as $d\to \infty$: 
	\begin{align*}
	\sup_{d \leq k \leq pN-p-d} \, \bigl | \la c_k^* J_{N,\Lambda} c_k \ra_{N,\mathcal{Z}} - \la \mathcal{L} \ra_{N,\mathcal{Z}}\, \la  c_k^*c_k  \ra_{N,\mathcal{Z}} \la \mathcal{R} \ra_{N,\mathcal{Z}} \bigr|
		 & \leq f(d), \\
	\bigl | \la J_{N,\Lambda}  \ra_{N,\mathcal{Z}} - \la \mathcal{L} \ra_{N,\mathcal{Z}}\,  \la \mathcal{R} \ra_{N,\mathcal{Z}} \bigr|
		 & \leq f(d).
	\end{align*}
\end{lemma}

\begin{proof}
Let  $d$ be large enough so that  $\eps_d \leq 1/4$ and $\delta_d \leq 1/4$. Let 
$c>0$ such that $\ln (1-x) \geq - c x$ when $|x|\leq 1/2$. For $pN - p \geq 3d$, we have 
\begin{equation*}
	\mathbf{1} \geq  \mathcal{B} \geq \mathbf{1} - c \sum_{j=d}^{pN-p-d} (\eps_j +\delta_{pN-p-j}) \hat n_j,
\end{equation*}
as an operator inequality. Noting that $\mathcal{B}$ and $c_k$, $c_k^*$ commute or anticommute 
with $\mathcal{L}$ and $\mathcal{R}$, 
\begin{align*}
	&  \bigl| \la c_k^*\, \mathcal{L\, B\, R}\, c_k \ra_{N,\mathcal{Z}} - \la c_k^* \, \mathcal{L\, R}\, c_k \ra_{N,\mathcal{Z}} \bigr| \\
		& \qquad  \leq \Bigl \la \mathcal{ L}  c_k^* (\mathbf{1}- \mathcal{B}) c_k \mathcal{L} \Bigr \ra_{N,\mathcal{Z}}^{1/2} \, 
			\Bigl \la \mathcal{R} c_k^* (\mathbf{1}- \mathcal{B}) c_k \mathcal{R}  \Bigr \ra_{N,\mathcal{Z}}^{1/2} \\
		& \qquad \leq \Bigl \la \Bigl( c_k^* (\mathbf{1} - \mathcal{B}) c_k \Bigr)^2 \Bigr \ra_{N,\mathcal{Z}}^{1/2} 
		   \leq \Bigl \la \Bigl(c \sum_{j=d}^{pN-p-d}(\eps_j + \delta_{pN-p-j}) \hat n_j \Bigr)^2 \Bigr \ra_{N,\mathcal{Z}}^{1/2}. 
\end{align*}
Using the uniform moment bounds 
from Prop.~\ref{prop:corr-bounds}, this can be further upper bounded by a constant times 
$	\sum_{j=d}^{\infty} (\eps_j  + \delta_{j}) $, 
which is finite and goes to $0$ 
as $d\to \infty$. 

Next, we want to use Prop.~\ref{prop:lattice-clust} in order to decouple $\mathcal{L}$, $c_k^*c_k = \hat n_k$ and $\mathcal{B}$. 
We note, first, that Prop.~\ref{prop:lattice-clust} is not directly applicable to $\mathcal{L}$ and $\mathcal{B}$, since these operators 
are not polynomials of 
creation and annihilation operators. Power series expansions yield a simple remedy. To illustrate the procedure, we explain 
how Prop.~\ref{prop:lattice-clust} can be applied to a product of two exponentials. 
We have, for suitable $K,D>0$ and $\varepsilon (d)\to 0$ as $d\to \infty$, 
\begin{align*}
	 & \Bigl|\bigl \la \exp( - s \hat n_0) \exp( - t \hat n_d)\bigr \ra_N - \bigl\la \exp( - s \hat n_0) \bigr\ra_N \bigl \la \exp( - t \hat n_d) \bigr\ra_N   \Bigr| \\
	& \qquad \leq K \sum_{k=0}^\infty \sum_{q=0}^\infty \frac{s^k}{k!} \frac{t^q}{q!} 
				D^{k+q} \varepsilon(d) \leq K \exp( s D) \exp(t D) \varepsilon(d).
\end{align*}  
Here we have used  Prop.~\ref{prop:corr-bounds} in the form $\la \hat n_0 ^{2k} \hat n_d^{2q} \ra_N^{1/2} \leq D^{k+q}$ for suitable $D$.
Something similar can be done to bound
\begin{equation} \label{eq:lrb}
 \la \mathcal{L}\, \hat n_k\, \mathcal{R} \ra_{N,\mathcal{Z}} - \la \mathcal{L}\ra_{N,\mathcal{Z}}\la  \hat n_k \ra_{N,\mathcal{Z}}\la \mathcal{R}\ra_{N,\mathcal{Z}}.
\end{equation}
The upper bound will involve 
\begin{equation*}
	\Bigl( \prod_{j=0}^{d-1} (1+ \epsilon_j + \delta_{pN - p - j}) \Bigr)^D 
\end{equation*}
and a similar term for the right boundary. We note that this term can be bounded, uniformly in $N$ and $d$, and deduce that 
the absolute value of~\eqref{eq:lrb} is bounded by some  function $g(d)$ with $g(d) \to 0$ as $d\to \infty$. 
In combination with our earlier bound which justified the replacement $\mathcal{B} \approx \mathbf{1}$,  
this proves the first inequality of the lemma when $d \geq d_0$ for some $d_0$. When $d\leq d_0$, 
we note that the left-hand side of the inequality can be bounded, uniformly in $N$, and set $f(d)$ equal to that bound.

The proof of the second inequality is similar. 
\end{proof}

Next, we observe that  $\la \mathcal{ L} \ra_{N,\mathcal{Z}}$ and $\la \mathcal{R} \ra_{N,\mathcal{Z}}$ stay bounded away from $0$ 
as $N\to \infty$.
 Indeed, if $(t_k)$  is a sequence of numbers in $[0,1]$ and $\sum_k t_k < \infty$, Jensen's inequality gives 
\begin{align*}
	\bigl \la \prod_{j=0}^{pN - p} (1 - t_j)^{\hat n_j} \bigr \ra_{N,\mathcal{Z}} &\geq 
			\exp\Bigl( \bigl \la \sum_{j=0}^{pN-p} \hat n_j \ln (1- t_j) \bigr \ra_{N,\mathcal{Z}} \Bigr) \\
			& \geq \exp\Bigl( K \sum_{j=0}^\infty  \ln(1-t_j) \Bigr) >0. 
\end{align*}
Here $K$ is a uniform upper bound for the occupation numbers $\la \hat n_k\ra_{N,\mathcal{Z}}$. 
This argument can be adapted without problems to  lower bound the expectations of $\mathcal{L}$ and $\mathcal{R}$. 
Note that, as operators with norm $\leq 1$, they have expectations upper bounded by $1$. 

As a consequence, we can pass to quotients and deduce from Lemma~\ref{lem:61}
\begin{equation*}
	\sup_{d \leq k \leq pN-p-d}\,  \Bigl| \frac{\la c_k^* J_{N,\Lambda} c_k \ra_{N,\mathcal{Z}}} {\la J_{N,\Lambda}  \ra_{N,\mathcal{Z}}}
			- \la c_k^* c_k \ra_{N,\mathcal{Z}} \Bigr|\leq g(d) 
\end{equation*}
for some $N$-independent function $g(d)$ which goes to $\infty$ as $d\to \infty$. 
From here the proof of Theorem~\ref{thm:thermolim} for the one-particle matrix for general $\Lambda$, i.e., the insensitivity to the domain of integration, 
is proven with the help of Eq.~\eqref{eq:one-part-finite} by imitating the proof of Theorem~\ref{thm:thermolim} for $\Lambda =\mathcal{Z}$ 
on p.~\pageref{proof:thermolim}. The proofs for general $n$-point functions are similar. 

\subsection{Symmetry breaking} \label{sec:symbreak}

We conclude the paper with a proof of Theorem~\ref{thm:symbreak}, which is essentially a consequence of results of~\cite{ajj}. Let us also recall that on thin cylinders, the slightly stronger statement that the one-particle density (and not just any correlation function) has a non-trivial period was proven in~\cite{jls}. 

\begin{proof}[Proof of Theorem~\ref{thm:symbreak}] 
The diagonal infinite volume correlation functions\\
 $\rho_n(z_1,\ldots,z_n;z_1,\ldots,z_n)$  
are the correlation functions (= factorial moment densities) of some point process on $\mathcal{Z}$. 
Because of Theorem~\ref{thm:thermolim}, the corresponding measure 
$P$ is the limit, in a suitable sense and up to shifts, of the measure $P_N$   with density 
$\propto |\Psi_N(z_1,\ldots,z_N) |^2$, choosing the finite cylinder $- p \gamma/2 \leq x \leq 
p(N-1/2)\gamma$ as the domain of integration. With this choice $P_N$ is exactly the Gibbs measure for $N$ particles moving in a neutralizing background, studied in~\cite{ajj}. 

Therefore, by~\cite[Theorem 3.1]{ajj}, if we shift $P$ by $\theta\in \R$ along the $x$-axis, we obtain a measure $P^\theta$ which is singular to $P$ unless $\theta$ is an integer multiple of $p\gamma$. 
Now, Theorem~\ref{thm:imbalance} together with Eq.~\eqref{eq:nak} shows that the point process 
satisfies conditions which ensure that it is uniquely determined by its correlation functions~\cite{dvj}. Remember that passing from correlation functions to the point process is like passing from moments of a probability measure to the measure itself.  Thus if $\theta$ is not an integer multiple of $p\gamma$, 
 the $\theta$-shifted measure  must have some correlation function which 
is different from the one for the original measure $P$. This proves Theorem~\ref{thm:symbreak}. 
\end{proof}

\begin{remark}
	Repeated shifts of the infinite-cylinder state $\omega_1(\cdot)=\la \cdot \ra$ by $\ell^2/R$ in the $x$-direction yield states $\omega_2,\ldots,\omega_{p}$. By Theorem~\ref{thm:symbreak}, those $p$ states are \emph{distinct}. They are actually also \emph{disjoint} (this notion generalizes mutual singularity of probability measures). This follows from general  arguments~\cite[Sect.4]{bratteli-robinson1}, combining the fact that $\omega_1,\ldots,  \omega_p$ are distinct and mixing (by Theorem~\ref{thm:clustering}), hence ergodic with respect to shifts in the $x$-direction.
\end{remark}
\vspace{1cm}

\textbf{Acknowledgments} This work was supported by the DFG Forschergruppe 718 ``Analysis and Stochastics in Complex
Physical Systems,'' and initiated during a stay in Princeton supported by 
 NSF grant PHY-0652854 and a Feodor Lynen research fellowship
of the Alexander von Humboldt-Stiftung. 



\newcommand{\etalchar}[1]{$^{#1}$}
\providecommand{\bysame}{\leavevmode\hbox to3em{\hrulefill}\thinspace}
\providecommand{\MR}{\relax\ifhmode\unskip\space\fi MR }
\providecommand{\MRhref}[2]{%
  \href{http://www.ams.org/mathscinet-getitem?mr=#1}{#2}
}
\providecommand{\href}[2]{#2}

\end{document}